\documentclass[a4paper,UKenglish]{lipics}
 
\usepackage{microtype}

\title{One-Counter Automata with Counter Observability}

\author{Benedikt Bollig}
\affil{LSV, ENS Cachan, CNRS, Inria, Universit{\'e} Paris-Saclay, France\\\texttt{bollig@lsv.fr}}

\authorrunning{B.\ Bollig}

\Copyright{B.\ Bollig}

\usepackage{xspace}
\usepackage{xcolor}
\usepackage{amsmath,amssymb}
\usepackage[pdflatex,recompilepics=false]{gastex}
\usepackage{enumerate}
\usepackage{enumitem}
\usepackage{extarrows}
\usepackage{colortbl}
\usepackage{upgreek}
\usepackage{booktabs}
\usepackage{pifont}
\usepackage{dsfont}
\usepackage{stmaryrd}
\usepackage{rotating}
\usepackage{mathrsfs}
\usepackage{cite}
\usepackage{graphicx,url}
\usepackage{cmll}
\usepackage{yfonts}
\usepackage{fancyhdr}

\theoremstyle{plain}

\newtheorem{claim}[theorem]{Claim}
\newtheorem{proposition}[theorem]{Proposition}

\makeatletter
\renewcommand{\paragraph}{\@startsection{paragraph}{6}{\z@}{2ex}{-0.7em}{\normalsize\bf}}
\makeatother

\usepackage{ifpdf}
\ifpdf
\usepackage[framemethod=TikZ]{mdframed}

\usepackage[colorinlistoftodos,prependcaption]{todonotes}
\fi

\newcounter{todocounter}

\gasset{frame=false,AHangle=30,AHlength=1.4,AHLength=1.6}

\gasset{AHangle=30,AHlength=1.4,AHLength=1.6}

\makeatletter
\newcommand{\xRightarrow}[2][]{\ext@arrow 0359\Rightarrowfill@{#1}{#2}}
\makeatother

\renewcommand{\epsilon}{\varepsilon}
\renewcommand{\phi}{\varphi}

\newcommand{\A}{\mathcal{A}}
\newcommand{\B}{\mathcal{B}}
\newcommand{\C}{\mathcal{C}}

\newcommand{\F}{\mathcal{B}}
\newcommand{\SA}{\mathcal{S}}

\newcommand{\init}{\iota}

\newcommand{\sem}[1]{\ensuremath{[\![#1]\!]}}

\newcommand{\true}{\mathit{true}}
\newcommand{\false}{\mathit{false}}

\newcommand{\N}{\mathds{N}}

\newcommand{\df}{:=}

\newcommand{\inc}{\raisebox{-0.1ex}{$\shneg$}}
\newcommand{\dec}{\raisebox{-0.2ex}{$\shpos$}}

\newcommand{\sinc}{\raisebox{-0.1ex}{$\scriptstyle\shneg$}}
\newcommand{\sdec}{\raisebox{-0.2ex}{$\scriptstyle\shpos$}}

\newcommand{\incmode}{\mathord{\nearrow}}
\newcommand{\decmode}{\mathord{\searrow}}
\newcommand{\nmode}{\mathord{\xrightarrow{~\;\,}}}

\newcommand{\kform}[1]{\pi_{#1}}

\newcommand{\srelp}[1]{\Longrightarrow_{\!#1}}
\newcommand{\srelpp}[2]{\stackrel{#2}{\Longrightarrow}_{\!#1}}

\newcommand{\Conf}{\mathit{Conf}}
\newcommand{\Confp}[1]{\mathit{Conf}_{\!\!#1}}

\newcommand{\cextr}{\mathit{oca}}
\newcommand{\vextr}{\mathit{vis}}
\newcommand{\oextr}{\mathit{obs}}
\newcommand{\extr}{\eta}

\newcommand{\cL}{{L}_\mathit{oca}}
\newcommand{\vL}{{L}_\mathit{vis}}
\newcommand{\oL}{L_\mathit{obs}}
\newcommand{\eL}{L_\extr}

\newcommand{\cv}{x}

\newcommand{\trans}{\Delta}

\newcommand{\OCO}{\text{OCA}\xspace}
\newcommand{\OCOs}{\text{OCAs}\xspace}

\newcommand{\ECAV}{\text{eOCA}\xspace}
\newcommand{\ECAVs}{\text{eOCAs}\xspace}

\newcommand{\set}[1]{[#1]}
\newcommand{\setz}[1]{[#1]_0}

\newcommand{\cR}{K}

\newcommand{\Op}{\mathsf{Op}}
\newcommand{\op}{\mathit{op}}

\newcommand{\consta}{c}
\newcommand{\constb}{d}
\newcommand{\lab}{\sigma}
\newcommand{\obsletter}{\tau}

\newcommand{\SetGuards}{\Omega}

\newcommand{\Guards}{\mathsf{Guards}}

\newcommand{\Guardsmod}{\Guards^\textup{mod}}

\newcommand{\URel}[3]{X^#3_{#1,#2}}

\newcommand{\pairs}[1]{\cR}

\newcommand{\cA}[1]{\bar{#1}}
\newcommand{\dA}[1]{#1^{\mathit{det}}}
\newcommand{\eA}[1]{#1^{\mathit{ext}}}

\newcommand{\Benc}{\B_\mathsf{enc}}

\newcommand{\thr}{m}

\begin{document}

\maketitle

\pagestyle{fancy}
\fancyhead{}
\renewcommand{\headrulewidth}{0pt}
\fancyfoot[C]{\vspace{2ex}\thepage}

\begin{abstract}
In a one-counter automaton (OCA), one can produce a letter from some finite alphabet, increment and decrement the counter by one, or compare it with constants up to some threshold. It is well-known that universality and language inclusion for OCAs are undecidable. In this paper, we consider OCAs with counter observability: Whenever the automaton produces a letter, it outputs the current counter value along with it. Hence, its language is now a set of words over an infinite alphabet. We show that universality and inclusion for that model are PSPACE-complete, thus no harder than the corresponding problems for finite automata. In fact, by establishing a link with visibly one-counter automata, we show that OCAs with counter observability are effectively determinizable and closed under all boolean operations.
\end{abstract}

\section{Introduction}

One-counter automata (OCAs) are a fundamental model of infinite-state systems. Their expressive power resides between finite automata and pushdown automata. Unlike finite automata, however, OCAs are not robust: They lack closure under complementation and have an undecidable universality, equivalence, and inclusion problem \cite{Greibach69,Ibarra79}.
Several directions to overcome this drawback have been taken. One may underapproximate the above decision problems in terms of bisimilarity \cite{Jancar2000} or overapproximate the system behavior by a finite-state abstraction, e.g., in terms of the downward closure or preserving the Parikh image \cite{vanLeeuwen1978,Parikh1966}.

In this paper, we consider a new and simple way of obtaining a robust model of one-counter systems. Whenever the automaton produces a letter from a finite alphabet $\Sigma$, it will also output the \emph{current} counter value along with it (transitions that manipulate the counter are not concerned). Hence, its language is henceforth a subset of $(\Sigma \times \N)^\ast$.
For obvious reasons, we call this variant \emph{OCAs with counter observability}. We will show that, under the observability semantics, OCAs form a robust automata model: They are closed under all boolean operations. Moreover, their universality and inclusion problem are in PSPACE and, as a simple reduction from universality for finite automata shows, PSPACE-complete.

These results may come as a surprise given that universality for OCAs is undecidable and introducing counter observability seems like an extension of OCAs. But, actually, the problem becomes quite different. The fact that a priori hidden details from a run (in terms of the counter values) are revealed makes the model more robust and the decision problems easier. Note that this is also what happens in input-driven/visibly pushdown automata \cite{Mehlhorn80,Alur2009} or their restriction of visibly OCAs \cite{BaranyLS06,S:CSL:06}. They all recognize languages over a \emph{finite} alphabet and the stack/counter operation can be deduced from the letter that is read. Interestingly, our proofs establish a link between the observability semantics and visibly OCAs, which somehow explains the robustness of OCAs under the observability semantics.

However, it is worth noting that revealing details from a system configuration does not always help, quite the contrary:
Though timed automata and counter automata are closely related \cite{HaaseOW16}, the universality problem of timed automata is decidable only if time stamps are excluded from the semantics \cite{AD94}.

Note that it is not only for the pure fact that we obtain a robust model that we consider counter observability. Counter values usually have a \emph{meaning}, such as energy level, value of a variable, or number of items yet to be produced (cf.\ Example~\ref{ex:job}). In those contexts, it is natural to include them in the semantics, just like including time stamps in timed automata.

\paragraph{Related Work.}

Apart from the connection with visibly OCAs, another model closely related to ours is that of \emph{strong automata} \cite{CzybaST15}. Strong automata operate on infinite alphabets and were introduced as an extension of \emph{symbolic automata} \cite{Bes08,DAntoniV14}.
Essentially, a transition of a strong automaton is labeled with a formula from monadic second-order (MSO) logic over some infinite structure, say $(\N,+1)$. In fact, the formula has two free first-order variables so that it defines a binary relation over $\N$. This relation is interpreted as a constraint between successive letters from the infinite alphabet. We will show that \OCOs with the observability semantics and strong automata over $(\N,+1)$ (extended by a component for the finite alphabet $\Sigma$) are expressively equivalent. This underpins a certain naturalness of the observability semantics. Note that the universality and the inclusion problem have been shown decidable for strong automata over $(\N,+1)$ \cite{CzybaST15}. However, strong automata do not allow us to derive any elementary complexity upper bounds. In fact, our model can be seen as an operational counterpart of strong automata over $(\N,+1)$.

\paragraph{Outline.}

Section~\ref{sec:cav} defines OCAs and their different semantics. Section~\ref{sec:det} relates the observability semantics with visibly OCAs and shows that, under the new semantics, \OCOs are closed under boolean operations and have a PSPACE-complete universality and inclusion problem. In Section~\ref{sec:strong-aut}, we show expressive equivalence of strong automata and OCAs with counter observability.
We conclude in Section~\ref{sec:conclusion}.
Omitted proofs or proof details can be found in the appendix.

\section{One-Counter Automata with Counter Observability}\label{sec:cav}

For $n \in \N = \{0,1,2,\ldots\}$, we set $\set{n} \df \{1,\ldots,n\}$ and $\setz{n} \df \{0,1,\ldots,n\}$. Given an alphabet $\Sigma$, the set of finite words over $\Sigma$,
including the empty word $\epsilon$, is denoted by $\Sigma^\ast$.

\subsection{One-Counter Automata and Their Semantics}

We consider ordinary one-counter automata over some nonempty finite alphabet $\Sigma$. In addition to a finite-state control and transitions that produce a letter from $\Sigma$, they have a counter that can be incremented, decremented, or tested for  values up to some threshold $\thr \in \N$ (as defined in \cite{BaranyLS06}). Accordingly, the set of \emph{counter operations} is $\Op = \{\inc,\dec\}$, where $\inc$ stands for ``increment the counter by one'' and $\dec$ for ``decrement the counter by one''. A transition is of the form $(q,k,\sigma,q')$ where $q,q'$ are states, $k \in \setz{m}$ is a counter test, and $\sigma \in \Sigma \cup \Op$. It leads from $q$ to $q'$, while $\sigma$ either produces a letter from $\Sigma$ or modifies the counter. However, the transition can only be taken if the current counter value $x \in \N$ satisfies
$k = \min\{\cv,\thr\}$. That is,
counter values can be checked against any number strictly below $\thr$ or for being at least $\thr$.
In particular, if $\thr = 1$, then we deal with the classical definition of one-counter automata, which only allows for zero and non-zero tests.

\begin{definition}[OCA, cf.\ \cite{BaranyLS06}]
A \emph{one-counter automaton} (or simply OCA) is a tuple $\A = (Q,\Sigma,\init,F,\thr,\Delta)$ where $Q$ is a finite set of \emph{states}, $\Sigma$ is a nonempty finite alphabet (disjoint from $\Op$), $\init \in Q$ is the \emph{initial state}, $F \subseteq Q$ is the set of \emph{final states}, $m \in \N$ is the \emph{threshold}, and $\Delta \subseteq Q \times \setz{m} \times (\Sigma \cup \Op) \times Q$ is the \emph{transition relation}.
We also say that $\A$ is an $\thr$-OCA.
Its size is defined as $|Q| + |\Sigma| + \thr + |\Delta|$.
\end{definition}

An OCA $\A = (Q,\Sigma,\init,F,\thr,\Delta)$ can have several different semantics:
\begin{itemize}
\item $\cL(\A) \subseteq \Sigma^\ast$ is the classical semantics when $\A$ is seen as an ordinary OCA.

\item $\vL(\A) \subseteq (\Sigma \cup \Op)^\ast$ is the \emph{visibly semantics} where, in addition to the letters from $\Sigma$, all counter movements are made apparent.

\item $\oL(\A) \subseteq (\Sigma \times \N)^\ast$ is the \emph{semantics with counter observability} where the current counter value is output each time a $\Sigma$-transition is taken.
\end{itemize}

We define all three semantics in one go. Let $\Confp{\A} \df Q \times \N$ be the set of \emph{configurations} of $\A$. In a configuration $(q,\cv)$, $q$ is the current state and $\cv$ is the current counter value. The \emph{initial} configuration is $(\init,0)$, and a configuration $(q,\cv)$ is \emph{final} if $q \in F$.

%\smallskip

\newcommand{\trace}{\mathit{trace}}

%Towards the language(s) of $\A$, 
We determine a global transition relation
${\srelp{\A}} \subseteq \Confp{\A} \times ((\Sigma \times \N) \cup \Op) \times \Confp{\A}$.
For two configurations $(q,\cv),(q',\cv') \in \Confp{\A}$ and $\obsletter \in (\Sigma \times \N) \cup \Op$, we have $(q,\cv) \srelpp{\A}{\obsletter\,} (q',\cv')$ if one of the following holds:
\begin{itemize}\itemsep=0.5ex
\item $\obsletter = \inc$ and $\cv' = \cv + 1$ and $(q,\min\{\cv,\thr\},\inc,q') \in \Delta$, or

\item $\obsletter = \dec$ and $\cv' = \cv - 1$ and $(q,\min\{\cv,\thr\},\dec,q') \in \Delta$, or

\item $\cv' = \cv$ and there is $a \in \Sigma$ such that $\obsletter = (a,x)$ and $(q,\min\{\cv,\thr\},a,q') \in \Delta$.
\end{itemize}

A \emph{partial run} of $\A$ is a sequence
$\rho = (q_0,\cv_0) \srelpp{\A}{\obsletter_1\;} (q_1,\cv_1) \srelpp{\A}{\obsletter_2\;} \ldots \srelpp{\A}{\obsletter_n\;} (q_n,\cv_n)$,
with $n \ge 0$.
%such that $(q_0,\cv_0)$ is the initial configuration.
If, in addition, $(q_0,\cv_0)$ is the initial configuration, then we say that $\rho$ is a \emph{run}.
We call $\rho$ \emph{accepting} if its last configuration $(q_n,\cv_n)$ is final.

Now, the semantics of $\A$ that we consider depends on what we would like to extract from $\trace(\rho) \df \obsletter_1 \ldots \obsletter_n \in ((\Sigma \times \N) \cup \Op)^\ast$. We let (given $(a,x) \in \Sigma \times \N$):
\begin{itemize}\itemsep=0.5ex
\item $\cextr((a,x)) = a$ and $\cextr(\inc) = \cextr(\dec) = \epsilon$

\item $\vextr((a,x)) = a$ and $\vextr(\inc) = \inc$ and $\vextr(\dec) = \dec$

\item $\oextr((a,x)) = (a,x)$ and $\oextr(\inc) = \oextr(\dec) = \epsilon$
\end{itemize}
Moreover, we extend each such mapping $\extr \in \{\cextr,\vextr,\oextr\}$ to $\obsletter_1 \ldots \obsletter_n \in ((\Sigma \times \N) \cup \Op)^\ast$ letting $\extr(\obsletter_1 \ldots \obsletter_n) \df \extr(\tau_1) \cdot \ldots \cdot \extr(\tau_n)$. Note that, hereby $u \cdot \epsilon = \epsilon \cdot u = u$ for any word $u$.

Finally, we let $\eL(\A) = \{\extr(\trace(\rho)) \mid \rho \text{ is an accepting run of } \A\}$.

\newcommand{\req}{\mathsf{req}}
\newcommand{\ack}{\mathsf{prod}}

\begin{figure}[t]
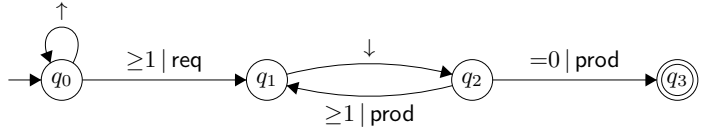

\begin{center}
\scalebox{0.9}{
\begin{gpicture}
\gasset{Nframe=y,Nw=5,Nh=5,Nmr=8,ilength=4,flength=4,AHangle=30} % circle

\node[Nmarks=i](q0)(0,0){$q_0$}
\node(q1)(25,0){$q_1$}
\node(q2)(50,0){$q_2$}
\node[Nmarks=r](q3)(75,0){$q_3$}

\drawedge(q0,q1){$\mathord{\ge} 1\,|\,\req$}
\drawloop[ELside=r,loopCW=n,loopdiam=4,loopangle=90](q0){$\inc$}
\drawedge[curvedepth=2,ELside=l](q1,q2){$\dec$}
\drawedge[curvedepth=2,ELside=l](q2,q1){$\mathord{\ge} 1\,|\,\ack$}
\drawedge(q2,q3){$\mathord{=}0\,|\,\ack$}

\end{gpicture}
}
\end{center}
\caption{A one-counter automaton with threshold $\thr = 1$\label{fig:cav}}
\end{figure}

\begin{example}\label{ex:job}
Consider the 1-OCA $\A$ from Figure~\ref{fig:cav} over $\Sigma=\{\req,\ack\}$.
For readability, counter tests $0$ and $1$ are written as $\mathord{=}0$ and $\mathord{\ge} 1$, respectively.
A transition without counter test stands for two distinct transitions, one for $\mathord{=}0$ and one for $\mathord{\ge} 1$ (i.e., the counter value may actually be arbitrary).
When looking at the semantics $\oL(\A)$, i.e., with counter observability, we can think of $(\req,n)$ signalizing that the production of $n \ge 1$ items is required (where $n$ is the current counter value). Moreover, $\ack$ indicates that an item has been produced so that, along a run, the counter value represents the number of items yet to be produced. It is thus natural to include it in the semantics. Concretely, we have
\begin{itemize}\itemsep=0.5ex
\item $\cL(\A) = \{\req\; \ack^n \mid n \ge 1\}$
\item $\vL(\A) = \{\inc^n \req\, (\dec\, \ack)^n \mid n \ge 1\}$
\item $\oL(\A) = \{(\req,n)(\ack,n-1)(\ack,n-2) \ldots (\ack,0) \mid n \ge 1\}$
\end{itemize}
Apparently, $\vL(\A)$ and $\oL(\A)$ are the only meaningful semantics in the context described above.
\end{example}

\newcommand{\vpalphabet}{\Gamma}

\begin{remark}
Visibly OCAs \cite{BaranyLS06,S:CSL:06} usually allow for general input alphabets $\vpalphabet$, which are partitioned into $\vpalphabet = \vpalphabet_\textup{inc} \uplus \vpalphabet_\textup{dec} \uplus \vpalphabet_\textup{nop}$ so that every $\gamma \in \vpalphabet$ is associated with a unique counter operation (or ``no counter operation'' if $\gamma \in \vpalphabet_\textup{nop})$. In fact, we consider here (wrt.\ the visibly semantics) a particular case where $\vpalphabet = \Sigma \cup \Op$ with $\vpalphabet_\textup{inc} = \{\inc\}$, $\vpalphabet_\textup{dec} = \{\dec\}$, and $\vpalphabet_\textup{nop} = \Sigma$.
\end{remark}

\subsection{Standard Results for OCAs}

Let us recall some well-known results for classical OCAs and visibly OCAs. For $\extr \in \{\cextr,\vextr,\oextr\}$, the \emph{nonemptiness problem for OCAs} wrt.\ the $\extr$-semantics is defined as follows: Given an \OCO $\A$, do we have $\eL(\A) \neq \emptyset$\,? Of course, this reduces to a reachability problem that is independent of the actual choice of the semantics:

\begin{theorem}[\!\!\cite{Valiant1975}]\label{thm:emptiness}
The nonemptiness problem for OCAs is NL-complete, wrt.\ any of the three semantics.
\end{theorem}

However, the universality (and, therefore, inclusion) problem for classical OCAs is undecidable:

\begin{theorem}[\hspace{-0.05em}\cite{Greibach69,Ibarra79}]
The following problem is undecidable: Given an OCA $\A$
with alphabet $\Sigma$, do we have $\cL(\A) = \Sigma^\ast$\,?
\end{theorem}

In this paper, we show that universality and inclusion are decidable when considering counter observability. To do so, we make use of the theory of the visibly semantics. Concretely, we exploit determinizability as well as closure under complementation and intersection.
In fact, the following definition of determinism only makes sense for the visibly semantics, but we will see later that a subclass of deterministic OCAs gives a natural notion of determinism for the observability semantics as well.

\begin{definition}[deterministic OCA]\label{def:doca}
An OCA $\A = (Q,\Sigma,\init,F,\thr,\Delta)$ is called \emph{deterministic} (dOCA or $\thr$-dOCA) if, for all $(q,k,\lab) \in Q \times \setz{m} \times (\Sigma \cup \Op)$, there is exactly one $q' \in Q$ such that $(q,k,\lab,q') \in \Delta$. In that case, $\Delta$ represents a (total) function $\delta: Q \times \setz{m} \times (\Sigma \cup \Op) \to Q$ so that we rather consider $\A$ to be the tuple $(Q,\Sigma,\init,F,\thr,\delta)$.
\end{definition}

A powerset construction like for finite automata can be used to determinize OCAs wrt.\ the visibly semantics \cite{BaranyLS06}. That construction also preserves the two other semantics. However, Definition~\ref{def:doca} only guarantees uniqueness of runs for words from $(\Sigma \cup \Op)^\ast$. That is, for complementation, we have to restrict to the visibly semantics (cf.\ Lemma~\ref{lem:compl} below).

\begin{lemma}[cf.\ \cite{BaranyLS06}]\label{lem:det}
Let $\A$ be an $\thr$-OCA. There is an $\thr$-dOCA $\dA{\A}$ of exponential size such that $\eL(\dA{\A}) = \eL(\A)$ for all $\extr \in \{\cextr,\vextr,\oextr\}$.
\end{lemma}

A (visibly) dOCA can be easily complemented wrt.\ the set of \emph{well-formed words} $\mathsf{WF}_\Sigma \df \{w \in (\Sigma \cup \Op)^\ast \mid$ no prefix of $w$ contains more $\dec$'s than $\inc$'s$\}$. In fact, for all OCAs $\A$ with alphabet $\Sigma$, we have $\vL(\A) \subseteq \mathsf{WF}_\Sigma$.

\begin{lemma}\label{lem:compl}
Let $\A=(Q,\Sigma,\init,F,\thr,\delta)$ be a dOCA and define $\cA{\A}$ as the dOCA $(Q,\Sigma,\init,Q \setminus F,\thr,\delta)$. Then, $\vL(\cA{\A}) = \mathsf{WF}_\Sigma \setminus \vL(\A)$.
\end{lemma}

Finally, visibility of counter operations allows us to simulate two OCAs in sync by a straightforward product construction (cf.\ Appendix~\ref{app:prod}):

\begin{lemma}[cf.\ \cite{Alur2009}]\label{lem:prod}
Let $\A_1$ be an $\thr_1$-OCA and $\A_2$ be an $\thr_2$-OCA over the same alphabet. There is a $\max\{m_1,m_2\}$-OCA $\A_1 \times \A_2$ of polynomial size such that $\vL(\A_1 \times \A_2) = \vL(\A_1) \cap \vL(\A_2)$. Moreover, if $\A_1$ and $\A_2$ are deterministic, then so is $\A_1 \times \A_2$.
\end{lemma}

\newcommand{\modform}[2]{#1 + #2\N}
\newcommand{\Qsink}{Q_\mathsf{sink}}
\newcommand{\fmap}{f}

\begin{figure}[t]
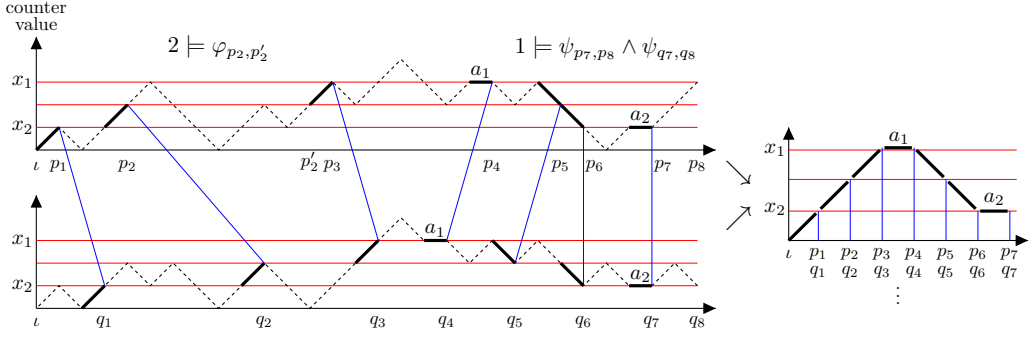

\begin{center}
%\scalebox{0.9}{
\scalebox{1}{
\!\!\!\!\!\!\!
\begin{gpicture}
\unitlength=0.3mm
\gasset{Nframe=n,Nw=4.8,Nh=4.8,Nmr=8,ilength=4,flength=4,AHlength=5.5,AHLength=6,AHangle=25} % circle

\node(cv)(0,63){\scalebox{0.7}{counter}}
\node(cv)(0,56){\scalebox{0.7}{value}}

\drawline[AHnb=1](0,0)(0,50)
\drawline[AHnb=1](0,0)(298,0)
\drawline[AHnb=0,linewidth=0.1,linecolor=red](0,10)(290,10)
\drawline[AHnb=0,linewidth=0.1,linecolor=red](0,20)(290,20)
\drawline[AHnb=0,linewidth=0.1,linecolor=red](0,30)(290,30)

\drawline[AHnb=0,dash={1.5}0](0,0)(10,10)(20,0)(30,10)(40,20)(50,30)(60,20)(70,10)(80,0)(90,10)(100,20)(110,10)(120,20)(130,30)(140,20)(150,30)(160,40)(170,30)(180,20)(190,30)(200,30)(210,20)(220,30)(230,20)(240,10)(250,0)(260,10)(270,10)(280,20)(290,30)

\drawline[AHnb=0,linewidth=1.4](0,0)(10,10)
\drawline[AHnb=0,linewidth=1.4](30,10)(40,20)
\drawline[AHnb=0,linewidth=1.4](120,20)(130,30)
\drawline[AHnb=0,linewidth=1.4](190,30)(200,30)
\drawline[AHnb=0,linewidth=1.4](220,30)(230,20)
\drawline[AHnb=0,linewidth=1.4](230,20)(240,10)
\drawline[AHnb=0,linewidth=1.4](260,10)(270,10)

\node(a)(195,34.5){\scalebox{0.8}{$a_1$}}
\node(a)(265,14.7){\scalebox{0.8}{$a_2$}}

\node(x)(-6,30){\scalebox{0.8}{$x_1$}}
\node(x)(-6,10){\scalebox{0.8}{$x_2$}}

\node(q)(0,-7){\scalebox{0.7}{$\init$}}
\node(q)(10,-7){\scalebox{0.7}{$p_1$}}
\node(q)(40,-7){\scalebox{0.7}{$p_2$}}
\node(q)(120,-5){\scalebox{0.7}{$p_2'$}}
\node(q)(130,-7){\scalebox{0.7}{$p_3$}}
\node(q)(200,-7){\scalebox{0.7}{$p_4$}}
\node(q)(230,-7){\scalebox{0.7}{$p_5$}}
\node(q)(245,-7){\scalebox{0.7}{$p_6$}}
\node(q)(275,-7){\scalebox{0.7}{$p_7$}}
\node(q)(290,-7){\scalebox{0.7}{$p_8$}}

\node(f)(80,45){\scalebox{0.85}{$2 \models \phi_{p_2,p_2'}$}}

\node(f)(250,45){\scalebox{0.85}{$1 \models \psi_{p_7,p_8} \wedge \psi_{q_7,q_8}$}}

\put(0,-70){
\drawline[AHnb=1](0,0)(0,50)
\drawline[AHnb=1](0,0)(298,0)
\drawline[AHnb=0,linewidth=0.1,linecolor=red](0,10)(290,10)
\drawline[AHnb=0,linewidth=0.1,linecolor=red](0,20)(290,20)
\drawline[AHnb=0,linewidth=0.1,linecolor=red](0,30)(290,30)

\drawline[AHnb=0,dash={1.5}0](0,0)(10,10)(20,0)(30,10)(40,20)(50,10)(60,20)(70,10)(80,0)(90,10)(100,20)(110,10)(120,0)(130,10)(140,20)(150,30)(160,40)(170,30)(180,30)(190,20)(200,30)(210,20)(220,30)(230,20)(240,10)(250,20)(260,10)(270,10)(280,20)(290,10)

\drawline[AHnb=0,linewidth=1.4](20,0)(30,10)
\drawline[AHnb=0,linewidth=1.4](90,10)(100,20)
\drawline[AHnb=0,linewidth=1.4](140,20)(150,30)
\drawline[AHnb=0,linewidth=1.4](170,30)(180,30)
\drawline[AHnb=0,linewidth=1.4](200,30)(210,20)
\drawline[AHnb=0,linewidth=1.4](230,20)(240,10)
\drawline[AHnb=0,linewidth=1.4](260,10)(270,10)

\node(a)(175,34.5){\scalebox{0.8}{$a_1$}}
\node(a)(265,14.5){\scalebox{0.8}{$a_2$}}

\node(x)(-6,30){\scalebox{0.8}{$x_1$}}
\node(x)(-6,10){\scalebox{0.8}{$x_2$}}

\node(q)(0,-6){\scalebox{0.7}{$\init$}}
%\node(q)(20,-6){\scalebox{0.7}{$q_0$}}
\node(q)(30,-6){\scalebox{0.7}{$q_1$}}
%\node(q)(90,-5){\scalebox{0.7}{$q_1'$}}
\node(q)(100,-6){\scalebox{0.7}{$q_2$}}
\node(q)(150,-6){\scalebox{0.7}{$q_3$}}
\node(q)(180,-6){\scalebox{0.7}{$q_4$}}
\node(q)(210,-6){\scalebox{0.7}{$q_5$}}
\node(q)(240,-6){\scalebox{0.7}{$q_6$}}
\node(q)(270,-6){\scalebox{0.7}{$q_7$}}
\node(q)(290,-6){\scalebox{0.7}{$q_8$}}
}

\put(330,-40){
\drawline[AHnb=1](0,0)(0,50)
\drawline[AHnb=1](0,0)(105,0)
\drawline[AHnb=0,linewidth=0.1,linecolor=red](0,13)(100,13)
\drawline[AHnb=0,linewidth=0.1,linecolor=red](0,27)(100,27)
\drawline[AHnb=0,linewidth=0.1,linecolor=red](0,40)(100,40)

\drawline[AHnb=0,linewidth=1.4](0,0)(12,12)
\drawline[AHnb=0,linewidth=1.4](14,14)(26,26)
\drawline[AHnb=0,linewidth=1.4](28,28)(40,40)
\drawline[AHnb=0,linewidth=1.4](42,41)(54,41)
\drawline[AHnb=0,linewidth=1.4](56,40)(68,28)
\drawline[AHnb=0,linewidth=1.4](70,26)(82,14)
\drawline[AHnb=0,linewidth=1.4](84,13)(96,13)

\node(a)(48.5,45.5){\scalebox{0.8}{$a_1$}}
\node(a)(90.2,18.3){\scalebox{0.8}{$a_2$}}

\node(x)(-6,40){\scalebox{0.8}{$x_1$}}
\node(x)(-6,14){\scalebox{0.8}{$x_2$}}

\node(vdots)(48,-22){\scalebox{0.7}{$\vdots$}}

\node(q)(0,-6){\scalebox{0.7}{$\init$}}
\node(q)(13,-6){\scalebox{0.7}{$p_1$}}
\node(q)(27,-6){\scalebox{0.7}{$p_2$}}
\node(q)(41,-6){\scalebox{0.7}{$p_3$}}
\node(q)(55,-6){\scalebox{0.7}{$p_4$}}
\node(q)(69,-6){\scalebox{0.7}{$p_5$}}
\node(q)(83,-6){\scalebox{0.7}{$p_6$}}
\node(q)(97,-6){\scalebox{0.7}{$p_7$}}

\node(q)(13,-14){\scalebox{0.7}{$q_1$}}
\node(q)(27,-14){\scalebox{0.7}{$q_2$}}
\node(q)(41,-14){\scalebox{0.7}{$q_3$}}
\node(q)(55,-14){\scalebox{0.7}{$q_4$}}
\node(q)(69,-14){\scalebox{0.7}{$q_5$}}
\node(q)(83,-14){\scalebox{0.7}{$q_6$}}
\node(q)(97,-14){\scalebox{0.7}{$q_7$}}

\node(arrow)(-22,30){\scalebox{1}{$\searrow$}}
\node(arrow)(-22,10){\scalebox{1}{$\nearrow$}}

\drawline[AHnb=0,linewidth=0.1,linecolor=blue](13,13)(13,0)
\drawline[AHnb=0,linewidth=0.1,linecolor=blue](27,27)(27,0)
\drawline[AHnb=0,linewidth=0.1,linecolor=blue](41,41)(41,0)
\drawline[AHnb=0,linewidth=0.1,linecolor=blue](55,41)(55,0)
\drawline[AHnb=0,linewidth=0.1,linecolor=blue](69,27)(69,0)
\drawline[AHnb=0,linewidth=0.1,linecolor=blue](83,13)(83,0)
\drawline[AHnb=0,linewidth=0.1,linecolor=blue](97,13)(97,0)
}

\drawline[AHnb=0,linewidth=0.1,linecolor=blue](10,10)(30,-60)
\drawline[AHnb=0,linewidth=0.1,linecolor=blue](40,20)(100,-50)
\drawline[AHnb=0,linewidth=0.1,linecolor=blue](130,30)(150,-40)
\drawline[AHnb=0,linewidth=0.1,linecolor=blue](200,30)(180,-40)
\drawline[AHnb=0,linewidth=0.1,linecolor=blue](230,20)(210,-50)
\drawline[AHnb=0,linewidth=0.1,linecolor=blue](240,10)(240,-60)
\drawline[AHnb=0,linewidth=0.1,linecolor=blue](270,10)(270,-60)

\end{gpicture}
}
\end{center}
\caption{Decompositions of two runs on $(a_1,3)(a_2,1)$, and corresponding runs in normal form\label{fig:stairs}}
\end{figure}

\section{Determinizing and Complementing \OCOs}\label{sec:det}

In this section, we will show that, under the \emph{observability semantics}, OCAs are effectively closed under all boolean operations. The main ingredient of the proof is a determinization procedure, which we first describe informally.

Let $\A = (Q,\Sigma,\init,F,\thr,\Delta)$ be the \OCO to be determinized (wrt.\ the observability semantics). Moreover, let $w =(a_1,x_1) \ldots (a_n,x_n) \in (\Sigma \times \N)^\ast$. Every run $\rho$ of $\A$ such that $\oextr(\trace(\rho)) = w$ has to have reached the counter value $x_1$ by the time it reads the first letter $a_1$. In particular, it has to perform at least $x_1$ counter increments. In other words, we can identify $x_1$ transitions that lift the counter value from $0$ to $1$, from $1$ to $2$, and, finally, from $x_1-1$ to $x_1$, respectively, and that are separated by partial runs that ``oscillate'' around the current counter value but, at the end, return to their original level. Similarly, before reading the second letter $a_2$, $\A$ will perform $|x_2 - x_1|$-many \emph{identical} counter operations to reach $x_2$, again separated by some oscillation phases, and so on. This is illustrated on the left hand side of Figure~\ref{fig:stairs} for two runs on the word $(a_1,3)(a_2,1)$.

We will transform $\A$ into another automaton that decomposes a run into oscillations and increment/decrement/letter transitions, but, in fact, abstracts away oscillations. Thus, the automaton starts in an increasing mode and goes straight to the value $x_1$. Once it reads letter $a_1$, it may go into an increasing or decreasing mode, and so on. Observe that a run $\rho$ of this new automaton is in a sort of normal form as illustrated on the right hand side of Figure~\ref{fig:stairs}. The crux is that $\vextr(\trace(\rho))$ is a \emph{unique} encoding of $w$: Of course, it determines the counter values output when a letter is read; and it is unique, since it continues incrementing (decrementing, respectively) until a letter is read.
This normalization and encoding finally allows us to apply known results on visibly one-counter automata for determinization and complementation.

There is a little issue here, since the possibility of performing an oscillation leading from $p_2$ to $p_2'$ (cf.\ left hand side of Figure~\ref{fig:stairs}) depends on the current counter value. However, it was shown in \cite{GMWT-lics09} that the set of counter values allowing for such a shortcut can be described as a boolean combination of arithmetic progressions that can be computed in polynomial time. We will, therefore, work with an extended version of \OCOs that includes arithmetic-progression tests (but is no more expressive, as we show afterwards).

The outline of this section is as follows: We present extended OCAs in Section~\ref{sec:ecav} and the link between the observability and the visibly semantics in Section~\ref{sec:obs-vis}. In Section~\ref{sec:univ}, we solve the universality and inclusion problem for OCAs wrt.\ the observability semantics.

\subsection{Extended One-Counter Automata}\label{sec:ecav}

While OCAs can only test a counter value up to some threshold, \emph{extended \OCOs} have access to boolean combinations of modulo constraints.
The set $\Guardsmod$ is given by the grammar
$\phi ::= \modform{\consta}{\constb} \mid \neg\phi \mid \phi \wedge \phi \mid \phi \vee \phi$
where $\consta,\constb \in \N$. We call $\modform{\consta}{\constb}$ an \emph{arithmetic-progression formula} and assume that $c$ and $d$ are encoded in unary. 
For $\cv \in \N$ (a counter value), we define $\cv \models \modform{\consta}{\constb}$ if $\cv = \consta + \constb \cdot i$ for some $i \in \N$.
Thus, we may use
%${=}0$ as an abbreviation for $\modform{0}{0}$, and
$\true$ as an abbreviation for $\modform{0}{1}$.
The other formulas are interpreted as expected. Moreover, given $\phi \in \Guardsmod$, we set $\sem{\phi} \df \{\cv \in \N \mid \cv \models \phi\}$.

Before we introduce extended OCAs, we will state a lemma saying that the ``possibility'' of a shortcut in terms of an oscillation (see above) is definable in $\Guardsmod$. Let $\A=(Q,\Sigma,\init,F,\thr,\Delta)$ be an \OCO and $p,q \in Q$. By $\smash{\URel{p}{q}{\A}}$, we denote the set of natural numbers $\cv \in \N$ such that
$(p,\cv) \mathrel{(\srelpp{\A}{\sinc\;} \cup \srelpp{\A}{\sdec\;})^\ast} (q,\cv)$,
i.e., there is a partial run from $(p,\cv)$ to $(q,\cv)$ that performs only counter operations. Moreover, we define $Y_{p,q}^\A$ to be the set of natural numbers $\cv \in \N$ such that
$(p,\cv) \mathrel{(\srelpp{\A}{\sinc\;} \cup \srelpp{\A}{\sdec\;})^\ast} (q,\cv')$
for some $\cv' \in \N$. Note that $\smash{\URel{p}{q}{\A}} \subseteq Y_{p,q}^\A$.
The following result is due to \cite[Lemmas~6--9]{GMWT-lics09}:

\begin{lemma}[\hspace{-0.05em}\cite{GMWT-lics09}]\label{lem:rpq}
Let $\A=(Q,\Sigma,\init,F,\thr,\Delta)$ be an \OCO and $p,q \in Q$. We can compute, in polynomial time, guards $\phi_{p,q},\psi_{p,q} \in \Guardsmod$ such that $\sem{\phi_{p,q}} = \URel{p}{q}{\A}$ and $\sem{\psi_{p,q}} = Y_{p,q}^{\A}$. In particular, the constants in $\phi_{p,q}$ and $\psi_{p,q}$ are all polynomially bounded.
\end{lemma}

\begin{definition}[extended OCA]\label{def:eoco}
An \emph{extended OCA (eOCA)} is a tuple $\A=(Q,\Sigma,\init,\fmap,\Delta)$ where $Q,\Sigma,\init$ are like in an OCA, $\fmap: Q \to \Guardsmod$ is the \emph{acceptance condition}, and $\trans$ is the \emph{transition relation}: a \emph{finite} subset of $Q \times \Guardsmod \times (\Sigma \mathrel{\cup} \Op) \times Q$
\end{definition}

Runs and the languages $L_\extr(\A)$, with $\extr \in \{\cextr,\vextr,\oextr\}$, of an \ECAV $\A=(Q,\Sigma,\init,\fmap,\Delta)$ are defined very similarly to OCAs.
In fact, there are only two (slight) changes:
\begin{enumerate}\itemsep=0.5ex
\item The definition of ${\srelp{\A}} \subseteq \Confp{\A} \times ((\Sigma \times \N) \cup \Op) \times \Confp{\A}$ is now as follows: For $(q,\cv),(q',\cv') \in \Confp{\A}$ and $\obsletter \in (\Sigma \times \N) \cup \Op$, we have $(q,\cv) \srelpp{\A}{\obsletter\;} (q',\cv')$ if there is $\phi \in \Guardsmod$ such that $\cv \models \phi$ and one of the following holds:\vspace{0.5ex}
\begin{itemize}%\itemsep=0.5ex
\item $\obsletter = \inc$ and $\cv' = \cv + 1$ and $(q,\phi,\inc,q') \in \Delta$, or

\item $\obsletter = \dec$ and $\cv' = \cv - 1$ and $(q,\phi,\dec,q') \in \Delta$, or

\item $\cv' = \cv$ and there is $a \in \Sigma$ such that $\obsletter = (a,x)$ and $(q,\phi,a,q') \in \Delta$.
\end{itemize}
\item A run is now \emph{accepting} if its last configuration $(q,\cv)$ is such that $\cv \models \fmap(q)$.
\end{enumerate}

Apart from these modifications, the languages $\cL(\A)$, $\vL(\A)$, and $\oL(\A)$ are defined in exactly the same way as for OCAs.

\subsection{From OCAs with Counter Observability to Visibly OCAs}
\label{sec:obs-vis}

\newcommand{\enc}[1]{\mathsf{enc}(#1)}
\newcommand{\Encs}[1]{\mathsf{Enc}_{#1}}
\newcommand{\encmap}{\mathsf{enc}}
\newcommand{\eval}{\lambda}

To establish a link between the observability and the visibly semantics, we will encode a word $w = (a_1,x_1)(a_2,x_2) \ldots (a_n,x_n) \in (\Sigma \times \N)^\ast$ as a word $\enc{w} \in (\Sigma \cup \Op)^\ast$ as follows:
\[\enc{w} \df \inc^{x_1} a_1\, \textup{sign}(x_2-x_1)^{|x_2-x_1|}\, a_2 \,\ldots\, \textup{sign}(x_n-x_{n-1})^{|x_n - x_{n-1}|}\, a_n\]
where, for an integer $z \in \mathds{Z}$,
we let $\textup{sign}(z) = \inc$ if $z \ge 0$, and $\textup{sign}(z) = \dec$ if $z < 0$.
%\[
%\textup{sign}(z) = 
%\begin{cases}
%\inc & \text{if } z \ge 0\\
%\dec & \text{if } z < 0
%\end{cases}
%\]
For example, $\enc{\epsilon} = \epsilon$ and $\enc{(a,5)(b,2)(c,4)} = \inc^5 a \,\dec^3 b \,\inc^2 c$.
The mapping $\encmap$ is extended to sets $L \subseteq (\Sigma \times \N)^\ast$ by $\enc{L} = \{\enc{w} \mid w \in L\}$.
%Similarly, for $L \subseteq (\Sigma \times \N)^\ast$ by $\enc{L} = \{\enc{w} \mid w \in L\}$
Let $\Encs{\Sigma} \df \enc{(\Sigma \times \N)^\ast}$ denote the set of \emph{valid encodings}.
Note that $\encmap$ is a \emph{bijection} between $(\Sigma \times \N)^\ast$ and $\Encs{\Sigma}$ and that the latter is the set of words of the form $w = u_1 a_1 u_2 a_2 \ldots u_n a_n$ where $a_i \in \Sigma$ and $u_i \in \{\inc\}^\ast \cup \{\dec\}^\ast$ for all $i \in \{1,\ldots,n\}$, and such that no prefix of $w$ contains more $\dec$'s than $\inc$'s.

Obviously, there is a small dOCA whose visibly semantics is $\Encs{\Sigma}$ (see Appendix~\ref{app:benc}). It will be needed later for complementation of OCAs wrt.\ the observability semantics.

\begin{lemma}\label{lem:benc}
There is a $0$-dOCA $\Benc$ with only four states such that $\vL(\Benc) = \Encs{\Sigma}$.
\end{lemma}

In fact, there is a tight link between the visibly and the observability semantics of OCAs provided the visibly semantics is in a certain normal form (cf.\ Appendix~\ref{app:enc} for the proof):

\begin{lemma}\label{lem:enc}
Let $\A$ be an OCA with alphabet $\Sigma$ such that $\vL(\A) \subseteq \Encs{\Sigma}$.
Then, we have $\vL(\A) = \enc{\oL(\A)}$ and, equivalently, $\oL(\A) = \encmap^{-1}(\vL(\A))$.
%For all $w \in (\Sigma \times \N)^\ast$, we have $w \in \oL(\A)$ iff $\enc{w} \in \vL(\A)$.
\end{lemma}

%Finally, to solve the inclusion problem, we will also rely on the following result, which easily follows from Lemmas~\ref{lem:prod}, \ref{lem:eoca}, and \ref{thm:oca-to-doca}:

Lemmas~\ref{lem:prod} and \ref{lem:enc} imply the following closure property: %, which will later be exploited to solve the inclusion problem:

\begin{proposition}\label{prop:intersect}
Let $\A_1$ and $\A_2$ be OCAs over $\Sigma$ such that $\vL(\A_1) \subseteq \Encs{\Sigma}$ and $\vL(\A_2) \subseteq \Encs{\Sigma}$. Then, $\oL(\A_1 \times \A_2) = \oL(\A_1) \cap \oL(\A_2)$ (where $\A_1 \times \A_2$ is due to Lemma~\ref{lem:prod}).
\end{proposition}

The next lemma constitutes the main ingredient of the determinization procedure. It will eventually allow us to rely on OCAs whose visibly semantics consists only of valid encodings.

%For a given OCA $\A$, we construct an eOCA $\eA{\A}$ that decomposes the runs of $\A$ so that the encodings of all words in $\oL(\eA{\A}) = \oL(\A)$ are already in the visibly semantics of $\eA{\A}$.

\begin{lemma}\label{lem:eoca}
Let $\A$ be an OCA. We can compute, in polynomial time, an eOCA $\eA{\A}$ such that
$\oL(\eA{\A}) = \oL(\A)$ and, for all $w \in \oL(\eA{\A})$, we have $\enc{w} \in \vL(\eA{\A})$.
%= \encmap^{-1}(\vL(\eA{\A}) \cap \Encs{\Sigma})$. 
\end{lemma}

\begin{proof}
Suppose $\A = (Q,\Sigma,\init,F,\thr,\Delta)$ is the given OCA.
We first translate a simple ``threshold constraint'' into an arithmetic expression that can be used as a guard in the eOCA $\eA{\A}$: Let $\kform{\thr} = \modform{\thr}{1}$, and
$\kform{k} = \modform{k}{0}$ for all $k \in \{0,\ldots,\thr-1\}$.

We define $\eA{\A} = (Q,\Sigma,\init,\fmap,\Delta')$ as follows:
Essentially, $\eA{\A}$ simulates $\A$ so that it has the same state space. However, when $\A$ allows for a shortcut (oscillation) from state $p$ to state $q$ (which will be checked in terms of $\phi_{p,q}$ from Lemma~\ref{lem:rpq}) and there is a transition $(q,k,\sigma,q')$ of $\A$, then $\eA{\A}$ may perform $\sigma$ and go directly from $p$ to $q'$, provided $\kform{k}$ is satisfied as well.
Formally, the transition relation is given as
$\Delta' = \{(p,\phi_{p,q} \wedge \kform{k},\sigma,q') ~\mid~ p \in Q \text{ and } (q,k,\sigma,q') \in \Delta\}$.
%\[\begin{array}{rl}
%& \bigl\{(p,\phi_{p,q} \wedge \kform{k},\inc,q') ~\mid~ p \in Q \text{ and } (q,k,\inc,q') \in \Delta\bigr\}\\[1ex]
%\cup & \bigl\{(p,\phi_{p,q} \wedge \kform{k},\dec,q') ~\mid~ p \in Q \text{ and } (q,k,\dec,q') \in \Delta\bigr\}\\[1ex]
%\cup & \bigl\{(p,\phi_{p,q} \wedge \kform{k},a,q') ~\mid~ p \in Q \text{, } a \in \Sigma \text{, and } (q,k,a,q') \in \Delta\bigr\}
%\end{array}\]
Moreover, a configuration $(p,x)$ is ``final'' in $\eA{\A}$ if the current counter value  $x$ satisfies $\psi_{p,q}$ for some $q \in F$ (cf.\ Lemma~\ref{lem:rpq}).
That is, for all $p \in Q$, we let
$\fmap(p) = \bigvee_{q \in F} \psi_{p,q}$.

%\smallskip

The correctness proof of the construction can be found in Appendix~\ref{app:eoca}.
%The proof of $\oL(\eA{\A}) = \oL(\A) = \encmap^{-1}(\vL(\eA{\A}) \cap \Encs{\Sigma})$ can be found in Appendix~\ref{app:eoca}.
%\qed
\end{proof}

To transform an eOCA back into an ordinary OCA while determinizing it and preserving its observability semantics, we will need a dOCA that takes care of the modulo constraints:
% (cf.\ Lemma~\ref{thm:oca-to-doca}).

\begin{lemma}\label{lem:BC}
Let $\SetGuards \subseteq \Guardsmod$ be a nonempty finite set. Set $\thr_\SetGuards \df \max\{\consta \mid \modform{\consta}{\constb}$ is an atomic subformula of some $\phi \in \SetGuards\}+2$. There are a dOCA $\B_\SetGuards=(Q,\Sigma,\init,Q,\thr_\SetGuards,\delta)$ of exponential size and $\eval: Q \to 2^\SetGuards$ such that, for all $(q,x) \in \Conf_{\B_\SetGuards}$ and all runs of $\B_\SetGuards$ ending in $(q,x)$, we have
$\eval(q) = \{\phi \in \SetGuards \mid x \models \phi\}$.
%and all $\phi \in C$, we have $\phi \in \eval(q)$ iff $x \models \phi$.
\end{lemma}

\newcommand{\incmap}[2]{\delta_{#1,#2}^\inc}
\newcommand{\decmap}[2]{\delta_{#1,#2}^\dec}
\newcommand{\pmap}[3]{\delta_{#1,#2}^{#3}}

\begin{proof}
We sketch the idea. The complete construction is given in Appendix~\ref{app:BC}.

For every arithmetic-progression formula $\modform{\consta}{\constb}$ that occurs in $\SetGuards$ (for simplicity, let us assume $\constb \ge 1$), we introduce a state component $\{0,1,\ldots,\consta\} \times \{0,1,\ldots,\constb-1\}$.
Increasing the counter, we increment the first component until $c$ and then count modulo $d$ in the second. We proceed similarly when decreasing the counter. The current state $(x,y) \in \setz{\consta} \times \setz{\constb-1}$ will then tell us whether $\modform{\consta}{\constb}$ holds, namely iff $x = \consta$ and $y = 0$. Finally, the mapping $\eval$ evaluates a formula based on the outcome for its atomic subformulas. Note that the size of $\B_\SetGuards$ is exponential in the number of arithmetic-progression formulas that occur in $\SetGuards$.
\end{proof}

We will now apply Lemma~\ref{lem:BC} to transform an eOCA into a dOCA
(cf.\ also Lemma~\ref{lem:det}).

\begin{lemma}\label{thm:oca-to-doca}
Let $\A$ be an eOCA. We can compute, in exponential time, a dOCA $\A'$
(deterministic according to Definition~\ref{def:doca})
such that $\eL(\A') = \eL(\A)$ for all $\extr \in \{\cextr,\vextr,\oextr\}$.
\end{lemma}

\begin{proof}
Suppose $\A=(Q,\Sigma,\init,\fmap,\Delta)$ is the given eOCA.
Let $\SetGuards \subseteq \Guardsmod$ be the set of formulas that occur in $\Delta$ or $\fmap$, and let $\B_\SetGuards=(\hat Q,\Sigma,\hat\init,\hat Q,\thr_\SetGuards,\hat\delta)$ be the dOCA along with the function $\eval$ according to Lemma~\ref{lem:BC}.

We build the dOCA ${\A'} = (Q',\Sigma,{\init'},{F'},\thr_\SetGuards,{\delta'})$ as follows.
Essentially, we perform a simple powerset construction for $\A$. Moreover, to eliminate modulo guards, we run $\B_\SetGuards$ in parallel.
Thus, the set of states is $Q' = 2^Q \times \hat Q $, with initial state $\init' = (\{\init\},\hat\init)$ and set of final states $F' = \{(P,q) \in Q' \mid \fmap(p) \in \eval(q)$ for some $p \in P\}$. Finally, the transition function is given by
$\delta'((P,q),k,\sigma) = (P',\hat\delta(q,k,\sigma))$
where $P' = \bigl\{p' \mid (p,\phi,\sigma,p') \in \Delta \cap (P \times \eval(q) \times \{\sigma\} \times Q)\bigr\}$.

For the correctness proof, we refer to Appendix~\ref{app:oca-to-doca}.
\end{proof}

There is a ``nondeterministic version'' of Lemma~\ref{thm:oca-to-doca}, which does not perform a powerset construction but rather computes a nondeterministic OCA. The latter is then still of exponential size, but only wrt.\ to the number of arithmetic-progression formulas in $\A$. 

With Theorem~\ref{thm:emptiness}, it follows that nonemptiness for \ECAVs can be solved in PSPACE. We do not know if this upper bound is tight.

\medskip

Let $\A$ be a dOCA with alphabet $\Sigma$ and let $w \in (\Sigma \times \N)^\ast$. By $\rho_{\!\A}(w)$, we denote the unique run of $\A$ such that $\vextr(\trace(\rho_{\!\A}(w))) = \enc{w}$.

By the following observation, which follows directly from Lemma~\ref{lem:enc}, it is justified to call any dOCA $\A$ with $\vL(\A) \subseteq \Encs{\Sigma}$ \emph{deterministic wrt.\ the observability semantics}:

\begin{lemma}\label{rem:unique}
Let $\A$ be a dOCA such that $\vL(\A) \subseteq \Encs{\Sigma}$. For every word $w \in (\Sigma \times \N)^\ast$,
we have $w \in \oL({\A})$ iff $\rho_{\!\A}(w)$ is accepting.
\end{lemma}

Altogether, we obtain that OCAs are determinizable wrt.\ the observability semantics.

\begin{theorem}[determinizability]\label{lem:doca-compl}
Let $\A$ be an OCA over $\Sigma$.
We can compute, in exponential time, an $\thr$-dOCA ${\A'}$ (with $\thr$ only polynomial) such that
$\oL({\A'}) = \oL(\A)$ and $\vL(\A') \subseteq \Encs{\Sigma}$.
\end{theorem}

\begin{proof}
Let $\A$ be the given OCA.
We apply Lemmas~\ref{lem:eoca} and \ref{thm:oca-to-doca} to obtain a dOCA $\widetilde{\A}$ of exponential size such that $\oL(\widetilde\A) = \oL(\A)$
and, for all $w \in \oL(\widetilde\A)$, we have $\enc{w} \in \vL(\widetilde\A)$.
We set $\A' = \widetilde\A \mathrel{\times} \Benc$ (cf.\ Lemmas~\ref{lem:prod} and \ref{lem:benc})
and obtain $\oL({\A'}) = \oL(\A)$ and $\vL(\A') \subseteq \Encs{\Sigma}$.
\end{proof}

We conclude that OCAs are complementable wrt.\ the observability semantics:

\begin{theorem}[complementability]\label{thm:complobs}
Let $\A$ be an OCA with alphabet $\Sigma$. We can compute, in exponential time, a dOCA $\bar{\A}$ such that $\oL({\bar{\A}}) = (\Sigma \times \N)^\ast \setminus \oL(\A)$.
\end{theorem}

\begin{proof}
We first transform $\A$ into the dOCA
$\A' = \widetilde\A \mathrel{\times} \Benc$ according to (the proof of) Theorem~\ref{lem:doca-compl}.
Suppose $\widetilde\A = (Q,\Sigma,\init,F,\thr,\delta)$.
Then, we set
$\bar\A = (Q,\Sigma,\init,Q \setminus F,\thr,\delta) \mathrel{\times} \Benc$. Note that $\bar\A$ is indeed a dOCA and that $\vL(\bar\A) \subseteq \Encs{\Sigma}$.
For $w \in (\Sigma \times \N)^\ast$, we have:
\[
\begin{array}[b]{c}
w \in \oL(\bar\A) \stackrel{\textup{Lem.~\ref{rem:unique}}}{\Longleftrightarrow}  \rho_{\!\bar\A}(w) \text{ is accepting}
 \stackrel{\textup{($\ast$)}}{\Longleftrightarrow}  \rho_{\!\A'}(w) \text{ is not accepting} 
\stackrel{\textup{Lem.~\ref{rem:unique}}}{\Longleftrightarrow}  w \not\in \oL(\A')
\end{array}
\]
Equivalence ($\ast$) holds as $\rho_{\!\bar\A}(w)$ and $\rho_{\!\A'}(w)$ have the same projection to the $Q$-component.
\end{proof}

Determinization and complementation of \emph{extended} OCAs are a priori more expensive: Lemmas~\ref{lem:rpq} and \ref{lem:eoca} only apply to OCAs so that one has to go through Lemma~\ref{thm:oca-to-doca} first.

\subsection{Universality and Inclusion Problem wrt.\ Observability Semantics}
\label{sec:univ}

We are now ready to solve the universality and the inclusion problem for OCAs wrt.\ the observability semantics.
The \emph{universality problem} is defined as follows: Given an \OCO $\A$ over some alphabet $\Sigma$, do we have $\oL(\A) = (\Sigma \times \N)^\ast$\,? The \emph{inclusion problem} asks whether, given OCAs $\A_1$ and $\A_2$, we have $\oL(\A_1) \subseteq \oL(\A_2)$.

\begin{theorem}
The universality problem and the inclusion problem for \OCOs wrt.\ the observability semantics are PSPACE-complete. In both cases, PSPACE-hardness already holds when $|\Sigma|=1$.
\end{theorem}

\begin{proof}
To solve the universality problem for a given \OCO $\A = (Q,\Sigma,\init,F,\thr,\Delta)$ in (nondeterministic) polynomial space, we apply the construction from Theorem~\ref{thm:complobs} (and, in particular, Theorem~\ref{lem:doca-compl}) on the fly to obtain a dOCA $\bar\A$ such that $\oL({\bar{\A}}) = (\Sigma \times \N)^\ast \setminus \oL(\A)$. That is, we have to keep in memory a state of the form $(P,q,r)$, where $P \subseteq Q$, $q$ is the modulo-counting component (Lemma~\ref{lem:BC}), and $r$ is a state of $\Benc$ (Lemma~\ref{lem:benc}). In addition, we will maintain a component for the current counter value. In fact, the latter can be supposed to be polynomially bounded (cf.\ \cite{CCHPW-fossacs16} for a tight upper bound) in the size of $\bar\A$. The size of $\bar\A$ is exponential in the size of $\A$, and so the required information can be stored in polynomial space. To compute a successor state of $(P,q,r)$, we first guess an operation $\sigma \in \Sigma \cup \Op$. We then compute $(P',q')$ according to the proof of Lemma~\ref{thm:oca-to-doca}
and update $r$ to $r'$ according to the type of $\sigma$. Note that this takes polynomial time only, since the function $\eval$ as required in Lemma~\ref{lem:BC} can be computed on the fly. Finally, the algorithm outputs ``non-universal'' when we find a final state of $\bar\A$.

For the inclusion problem, we rely on Proposition~\ref{prop:intersect} and perform the determinization procedure on-the-fly for \emph{both} of the given \OCOs.

\smallskip

For the lower bound, we will restrict to the universality problem, since it is a special case of the inclusion problem. We reduce from the universality problem for ordinary finite automata, which is known to be PSPACE-complete \cite{MeyerS72}.
If we suppose that $\Sigma$ is part of the input, then there is a straightforward reduction, which essentially takes the (ordinary) finite automaton
and adds self-looping increment/decrement transitions to each state.
Assuming $|\Sigma|=1$, the reduction is as follows. Let $\A$ be a finite automaton over some finite alphabet $\Gamma = \{a_0,\ldots,a_{n-1}\}$. We construct an \OCO $\A'$ over the singleton alphabet $\Sigma$ such that $L(\A) = \Gamma^\ast$ iff $L(\A') = (\Sigma \times \N)^\ast$. The idea is to represent letter $a_i$ by (counter) value $i$. To obtain $\A'$, an $a_i$-transition in $\A$ is replaced with a gadget that nondeterministically outputs $i$ or any other natural number strictly greater than $n-1$.
\end{proof}

\section{Relation with Strong Automata}\label{sec:strong-aut}

\newcommand{\varx}{\mathsf{x}}
\newcommand{\vary}{\mathsf{x}'}
\newcommand{\varyy}{\mathsf{y}}
\newcommand{\marker}{\$}
\newcommand{\markerone}{\$_1}
\newcommand{\markertwo}{\$_2}
\newcommand{\enum}{\eta}
\newcommand{\transform}{\Phi}
\newcommand{\Forms}{\mathcal{F}}
\newcommand{\initform}{\Psi}
\newcommand{\incform}{\mu}
\newcommand{\mymatrix}[4]{\smash{[\begin{array}{cc}\scriptstyle{#3} &\! \scriptstyle{#4}\\[-1.6ex]\scriptstyle{#1} &\! \scriptstyle{#2}\end{array}]}}
\newcommand{\transformp}[4]{\smash{\transform[\begin{array}{cc}\scriptstyle{#3} &\! \scriptstyle{#4}\\[-1.6ex]\scriptstyle{#1} &\! \scriptstyle{#2}\end{array}]}}
\newcommand{\initformp}[2]{\smash{\initform[\begin{array}{c}\scriptstyle{#2}\\[-1.7ex]\scriptstyle{#1}\end{array}]}}

In this section, we show that OCAs with counter observability are expressively equivalent to strong automata over $(\N,+1)$ \cite{CzybaST15}. As the latter are descriptive in spirit, OCAs can thus be seen as their operational counterpart.

Let us first give a short account of monadic-second order (MSO) logic over $(\N,+1)$ (see \cite{Tho97handbook} for more details). We have infinite supplies of first-order variables, ranging over $\N$, and second-order variables, ranging over subsets of $\N$. The atomic formulas are $\true$, $\vary = \varx + 1$, $\vary = \varx$, and $\varx \in \mathsf{X}$ where $\varx$ and $\vary$ are first-order variables and $\mathsf{X}$ is a second-order variable. Those formulas have the expected meaning. Further, MSO logic includes all boolean combinations, first-order quantification $\exists \varx \transform$, and second-order quantification $\exists \mathsf{X} \transform$. The latter requires that there is a (possibly infinite) subset of $\N$ satisfying $\transform$. As abbreviations, we may also employ $\vary = \varx - 1$ and formulas of the form $\varx' \in (\varx + c + d\N)$, where $c,d \in \N$. This does not change the expressive power of MSO logic. 

In the following, we assume that $\varx$ and $\vary$ are two distinguished first-order variables. We write $\transform(\varx,\vary)$ to indicate that the free variables of $\transform$ are among $\varx$ and $\vary$. If $\transform(\varx,\vary)$ is evaluated to true when $\varx$ is interpreted as $x \in \N$ and $\vary$ is interpreted as $x' \in \N$, then we write $(x,x') \models \transform$.
In fact, a transition of a strong automaton is labeled with a formula $\transform(\varx,\vary)$ and can only be executed if $(x,x') \models \transform$ where $x$ and $x'$ are the natural numbers read at the previous and the current position, respectively. Thus, two successive natural numbers in a word can be related explicitly in terms of an MSO formula.

\begin{definition}[\!\!\cite{CzybaST15}]\label{def:strong-aut}
A \emph{strong automaton} is a tuple $\SA = (Q,\Sigma,\init,F,\Delta)$ where
$Q$ is the finite set of \emph{states}, $\Sigma$ is a nonempty finite alphabet, $\init \in Q$ is the \emph{initial state}, and $F \subseteq Q$ is the set of \emph{final states}. Further, $\Delta$ is a \emph{finite} set of \emph{transitions}, which are of the form $(q,\transform,a,q')$ where $q,q' \in Q$ are the source and the target state, $a \in \Sigma$, and $\transform(\varx,\vary)$ is an MSO formula.
\end{definition}

Similarly to an OCA, $\SA$ induces a relation
${\srelp{\,\SA}} \subseteq \mathit{Conf}_{\!\SA} \times (\Sigma \times \N) \times \mathit{Conf}_{\!\SA}$, where $\mathit{Conf}_{\!\SA} = Q \times \N$.
For $(q,x),(q',x') \in \mathit{Conf}_{\!\SA}$ and $(a,y) \in \Sigma \times \N$, we have
$(q,\cv) \xRightarrow{(a,y)\,}_{\SA} (q',\cv')$ if $y = x'$ and there is an MSO formula $\transform(\varx,\vary)$ such that $(q,\transform,a,q') \in \Delta$ and $(x,x') \models \transform$.
A \emph{run} of $\SA$ on $w = (a_1,x_1) \ldots (a_n,x_n) \in (\Sigma \times \N)^\ast$ is a sequence
$\rho = (q_0,x_0) \xRightarrow{(a_1,x_1)\,}_{\SA} (q_1,\cv_1) \xRightarrow{(a_2,x_2)\,}_{\SA} \ldots \xRightarrow{(a_n,x_n)\,}_{\SA} (q_n,\cv_n)$
such that $q_0 = \init$ and $x_0 = 0$. It is \emph{accepting} if $q_n \in F$.

The language $L(\SA) \subseteq (\Sigma \times \N)^\ast$ of $\SA$ is defined as the set of words $w \in (\Sigma \times \N)^\ast$ such that there is an accepting run of $\SA$ on $w$.

\begin{figure}[t]
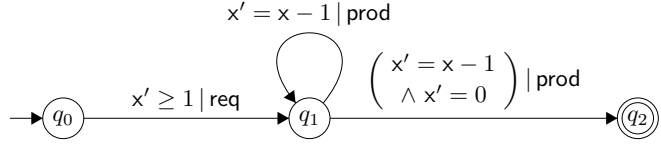

\begin{center}
\scalebox{0.9}{
\begin{gpicture}
\gasset{Nframe=y,Nw=5,Nh=5,Nmr=8,ilength=4,flength=4,AHangle=30} % circle

\node[Nmarks=i](q0)(0,0){$q_0$}
\node(q1)(30,0){$q_1$}
\node[Nmarks=r](q2)(70,0){$q_2$}

\drawedge(q0,q1){$\varx' \ge 1\,|\,\req$}
\drawloop[ELside=r,loopCW=n,loopdiam=8,loopangle=90](q1){$\varx' = \varx - 1\,|\,\ack$}
\drawedge[ELside=l](q1,q2){$\left(\!\begin{array}{c}\varx' = \varx - 1\\{\wedge}~\varx'=0\end{array}\!\right)|\,\ack$}

\end{gpicture}
}
\end{center}
\caption{A strong automaton over $(\N,+1)$\label{fig:strong-aut}}
\end{figure}

\begin{example}
We refer to the OCA $\A$ from Example~\ref{ex:job}. Figure~\ref{fig:strong-aut} depicts a strong automaton $\SA$ such that $L(\SA) = \oL(\A) = \{(\req,n)(\ack,n-1)(\ack,n-2) \ldots (\ack,0) \mid n \ge 1\}$.
\end{example}

In fact, we can transform any OCA into an equivalent strong automaton preserving the observability semantics, and vice versa:

\begin{theorem}\label{thm:strong-aut}
Let $L \subseteq (\Sigma \times \N)^\ast$.
There is an OCA $\A$ such that $\oL(\A) = L$ iff
there is a strong automaton $\SA$ such that $L(\SA) = L$.
\end{theorem}

\begin{proof}
``$\Longrightarrow$'': This is the easier direction. Using the following observation, we can directly transform an OCA into a strong automaton: For all states $q$ and $q'$ of a given OCA $\A$, there is an MSO formula $\transform_{q,q'}(\varx,\vary)$ such that, for all $x,x' \in \N$, we have $(x,x') \models \transform_{q,q'}$ iff $(q,\cv) \mathrel{\smash{\bigl(\srelpp{\A}{\sinc\;} \cup \srelpp{\A}{\sdec\;}\bigr)^\ast}} (q',\cv')$. The existence of $\transform_{q,q'}$ can be shown using a two-way automaton over infinite words \cite{Pecuchet85}, which simulates $\A$ and can be translated into an MSO formula \cite{Tho97handbook}.
We refer to Appendix~\ref{app:strong-aut} for more details.

\medskip

\noindent``$\Longleftarrow$'':
We will transform a strong automaton $\SA$ into an \emph{extended} OCA $\A$ such that $\oL(\A) = L(\SA)$. The main ingredient of $\A$ is a dOCA $\C_\transform$, with $\transform(\varx,\vary)$ an MSO formula, that, starting at 0, goes straight to two counter values $x$ and $x'$, and decides whether $(x,x') \models \transform$ or not:

\begin{claim}\label{cl:Cphi}
Let $\transform(\varx,\vary)$ be an MSO formula. There is a dOCA $\C_\transform$ over the alphabet $\{\markerone,\markertwo\}$ such that $\vL(\C_\transform) = {L}_1 \cup {L}_2$ where
\begin{itemize}\itemsep=0.5ex
\item ${L}_1 = \{\inc^x \markerone \inc^y \markertwo \mid x,y \in \N \text{ such that }  (x,x+y) \models \transform\}$ and

\item ${L}_2 = \{\inc^{x+y} \markerone \dec^y \markertwo \mid x,y \in \N \text{ such that } (x+y,x) \models \transform\}$\,.
\end{itemize}
\end{claim}

\noindent\textit{Proof of Claim~\ref{cl:Cphi}.}~
By Lemma~\ref{lem:det}, it is actually sufficient to come up with a \emph{nondeterministic} OCA. Using B{\"u}chi's theorem (cf.\ \cite{Tho97handbook}), we first transform $\transform(\varx,\vary)$ into
\emph{finite automata} $\F_1$ and $\F_2$
recognizing the following (regular) languages over the alphabet $\{\markerone,\markertwo,\Box\}$:
\begin{itemize}\itemsep=0.5ex
\item $L(\F_1) = \{\Box^x \markerone \Box^{y} \markertwo  \mid x,y \in \N \text{ such that } (x,x + y) \models \transform\}$

\item $L(\F_2) = \{\Box^x \markertwo \Box^{y} \markerone  \mid x,y \in \N \text{ such that } (x+y,x) \models \transform\}$
\end{itemize}
OCAs are easily seen to be closed under union (wrt.\ any semantics). So, it is enough to translate $\B_1$ and $\B_2$ into OCAs $\C_1$ and $\C_2$ such that $\vL(\C_1) = L_1$ and $\vL(\C_2) = L_2$.
While the construction of $\C_1$ is immediate, it is less obvious how to transform $\B_2$ into a suitable $\C_2$.

Note that $L(\B_2)$ can be written as the finite union of sets $\{\Box^{x} \markertwo \Box^{y} \markerone \mid x \in \sem{\modform{c_1}{d_1}}$ and $y \in \sem{\modform{c_2}{d_2}}\}$ with $c_1,d_1,c_2,d_2 \in \N$.
This is achieved by determinizing $\F_2$ and splitting it into components as illustrated in the upper part of Figure~\ref{fig:Bphi}. As we can of course handle finite unions, we assume that $\F_2$ is of that form.
For simplicity, we also suppose $d_1,d_2 \ge 1$.

\begin{figure}[t]
\begin{center}
\scalebox{0.9}{
\begin{gpicture}
  \gasset{Nframe=y,Nw=4,Nh=4,Nmr=8}
  \unitlength=0.8mm
  \node[Nframe=n](n3)(-13,0){$\F_2$}
  \node[iangle=180,Nmarks=i](q0)(0,0){}
  \node(q1)(25,0){}
  \node(q2)(50,0){}
  \node(q3)(75,0){}
  \node[Nmarks=r](q4)(100,0){}    
  \drawedge(q0,q1){$\Box^{c_1}$}
  \drawedge(q1,q2){$\markertwo$}
  \drawedge(q2,q3){$\Box^{c_2}$}
  \drawedge(q3,q4){$\markerone$}
  \drawloop[loopCW=n,loopdiam=6,loopangle=90,ELside=r](q1){~~$\Box^{d_1}$}
  \drawloop[loopCW=n,loopdiam=6,loopangle=90,ELside=r](q3){~~$\Box^{d_2}$}

\put(0,-15){
  \node[Nframe=n](n3)(-13,0){$\C_2$}
  \node[iangle=180,Nmarks=i](q0)(0,0){}
  \node(q1)(25,0){}
  \node(q2)(50,0){}
  \node(q3)(75,0){}
  \node(q4)(100,0){}    
  \drawedge(q0,q1){${\inc}^{c_1}$}
  \drawedge(q1,q2){$\epsilon$}
  \drawedge(q2,q3){${\inc}^{c_2}$}
  \drawedge(q3,q4){$\markerone$}
  \drawloop[loopCW=y,loopdiam=6,loopangle=-90,ELside=l](q1){~~${\inc}^{d_1}$}
  \drawloop[loopCW=y,loopdiam=6,loopangle=-90,ELside=l](q3){~~${\inc}^{d_2}$}

  \node(q5)(125,0){}
  \node[Nmarks=r](q6)(150,0){}      
  \drawedge(q4,q5){${\dec}^{c_2}$}
  \drawedge(q5,q6){$\markertwo$}
  \drawloop[loopCW=y,loopdiam=6,loopangle=-90,ELside=l](q5){~~${\dec}^{d_2}$}
  
  \node[Nframe=n](q5)(85,-9){$\left.\begin{array}{c}~\\[0.4ex]~\end{array}\right\}\scriptstyle{n_2\textup{-times}}$}
  \node[Nframe=n](q5)(115,-9){$\scriptstyle{n_2'\textup{-times}}\left\{\begin{array}{c}~\\[0.4ex]~\end{array}\right.$}
  \node[Nframe=n](n3)(100,-14){$\scriptstyle{n_2 \;\mathrel{\text{mod}}\; d_1 \,=\, n_2' \;\mathrel{\text{mod}}\; d_1}$}
  \node[Nframe=n](n4)(115,8){$\overbrace{~~~~~~~~~~~~~~~~~}^{\textup{counter value} \;\ge\; c_1}$}
}
\end{gpicture}
}
\end{center}
\caption{Transformation of $\B_2$ into $\C_2$\label{fig:Bphi}}
\end{figure}

From $\F_2$, we obtain the $c_1$-OCA $\C_2$ for $L_2$ depicted in the lower part of Figure~\ref{fig:Bphi}. It works as follows: It first produces $\inc^{x+y}\markerone$, where $x = c_1 + d_1n_1$ and $y = c_2 + d_2n_2$ for some $n_1,n_2 \in \N$. In doing so, it counts, modulo $d_1$, how often it traverses the ${\inc}^{d_2}$-loop. That is, after reading $\markerone$, it will be aware of $k \df n_2 \mathrel{\text{mod}} d_1$. It then continues reading $\dec^{z} \markertwo$ where $z = c_2 + d_2n_2'$ for some $n_2' \in \N$ such that $k = n_2' \mathrel{\text{mod}} d_1$ and $x + y - z \ge c_1$ (recall that we define a $c_1$-OCA).

We claim that
$(x+y,x+y-z) \models \transform$, which shows $\vL(\C_2) \subseteq L_2$. Since $n_2 \mathrel{\text{mod}} d_1 = n_2' \mathrel{\text{mod}} d_1 =: k$, we have $n_2 = m_1 d_1 + k$ and $n_2' = m_2 d_1 + k$ for some $m_1,m_2 \in \N$. Thus, 
\[\begin{array}{rl}
x + y - z\!\!\!\!
& =\, c_1 + d_1n_1 + c_2 + d_2n_2 - (c_2 + d_2n_2')\\
& =\, c_1 + d_1n_1 + c_2 + d_2(m_1 d_1 + k) - c_2 - d_2(m_2 d_1 + k)\\
& =\, c_1 + d_1n_1 + d_2 m_1 d_1 - d_2 m_2 d_1\\
& =\, c_1 + d_1(n_1 + d_2 m_1 - d_2 m_2) \;\in\; \sem{\modform{c_1}{d_1}}\,.
\end{array}
\]
Note that membership in $\sem{\modform{c_1}{d_1}}$ holds due to $x + y - z \ge c_1$.
We obtain $\Box^{x + y - z} \markertwo \Box^{z} \markerone \in L(\F_2)$ and, therefore, $(x + y,x + y - z) \models \Phi$.

It remains to prove $L_2 \subseteq \vL(\C_2)$.
Suppose $x,y \in \N$ such that $(x+y,x) \models \transform$. We have $\Box^{x} \markertwo \Box^{y} \markerone \in L(\F_2)$. Thus, $x \in \sem{\modform{c_1}{d_1}}$ and $y \in \sem{\modform{c_2}{d_2}}$. But this implies $\inc^{x+y} \markerone \dec^y \markertwo \in \vL(\C_2)$: The automaton $\C_2$ takes the same number of $d_2$-loops before and after reading $\markerone$, namely $(y-c_2)/d_2$ many. \qed$_{\textup{Claim~\ref{cl:Cphi}}}$

\medskip

We now turn to the actual translation of a strong automaton $\SA = (Q,\Sigma,\init,F,\Delta)$ into an eOCA $\A$ such that $\oL(\A) = L(\SA)$. Let $\Forms$ denote the set of MSO formulas that occur in the transitions of $\SA$. Without loss of generality, we assume that $\Forms$ contains $\false \df \neg\true$.

For all $\transform \in \Forms$, we compute the dOCA $\C_\transform = (Q_\transform,\{\markerone,\markertwo\},\init_\transform,F_\transform,m_\transform,\delta_\transform)$ according to Claim~\ref{cl:Cphi}.
Without loss of generality, we assume that there is $m \in \N$ such that $m = m_\transform$ for all $\transform \in \Forms$. That is, $\delta_\transform$ is of the form $Q_\transform \times \setz{m} \times (\{\markerone,\markertwo\} \mathrel{\cup} \Op) \to Q_\transform$ (Definition~\ref{def:doca}). We will also assume that the sets $(Q_\transform)_{\transform \in \Forms}$ are pairwise disjoint.
For $r \in Q_\transform$, we can easily compute a formula $\incform_r \in \Guardsmod$ such that $\sem{\incform_r} =
\{x \in \N \mid (\init_\transform,0) \mathrel{\bigl(\srelpp{\,\C_\transform}{\sinc\;}\bigr)^x} (r,x)\}$.

We now define the eOCA $\A = (Q',\Sigma,I,\fmap,\Delta')$ satisfying $\oL(\A) = L(\SA)$. Note that we use a set of initial states $I \subseteq Q'$ instead of one initial state; it is easily seen that this extension does not add expressive power to eOCAs. Whenever $\A$ outputs a counter value $x$ (and at the very beginning of an execution), it guesses (i) a formula $\transform$ that will label the next transition of $\SA$ to be simulated, as well as (ii) an \emph{entrance point} in $\C_\transform$. The latter is a state $r$ of $\C_\transform$ that can be reached from $\init_\transform$ by reading $\inc^x$ (which will be guaranteed by the guard $\incform_r$). Then, $\A$ will continue simulating $\C_\transform$ until another counter value $x'$ is output, so that $\A$ can determine whether $(x,x') \models \transform$.

Accordingly, a state of $\A$ is a tuple $(p,\transform,q) \in Q' = Q \times \Forms \times \bigl(\bigcup_{\transform \in \Forms} Q_\transform\bigr)$, where $p$ represents the current state of $\SA$, $\transform$ is the guessed next transition formula, and $q$ is the current state of $\C_\transform$.
The set of initial states is $I = \{(\init,\transform,\delta_\transform(\init_\transform,0,\markerone)) \mid \transform \in \Forms\}$. The guess $\transform = \false$ is used to signalize that no further transition will be taken in $\SA$. Thus, the acceptance condition $\fmap$ maps every state $(p,\transform,q)$ such that $p \in F$ and $\transform = \false$ to $\true$, and all other states to $\neg\true$. It remains to define the transition relation $\Delta'$:
\begin{itemize}\itemsep=1ex
\item $\Sigma$-transitions:  For all $(p,\transform,a,p') \in \Delta$ with $\transform \neq \false$, all $k \in \setz{m}$ (recall that $\C_\transform$ is an $m$-dOCA), all $q \in  Q_{\transform}$ such that $\delta_\transform(q,k,\markertwo) \in F_\transform$, all $\transform' \in \Forms$, and all $r \in Q_{\transform'}$, we have a transition
$(p,\transform,q) \xrightarrow{\incform_r \wedge \kform{k}\,|\, a\,} (p',\transform',\delta_{\transform'}(r,k,\marker_1))$, with $\kform{k}$ as in the proof of Lemma~\ref{lem:eoca}.

\item $\Op$-transitions: For all $(p,\transform,q) \in Q'$ with $\transform \neq \false$, all $k \in \setz{m}$, and all $\op \in \Op$, we have a transition
$(p,\transform,q) \xrightarrow{\kform{k}\,|\,\op\,} (p,\transform,\delta_\transform(q,k,\op))$.
\end{itemize}
The proof of $\oL(\A) = L(\SA)$ can be found in Appendix~\ref{app:strong-aut}.
\end{proof}

\section{Conclusion}\label{sec:conclusion}

The observability semantics opens several directions for follow-up work.
We may carry it over to other classes of infinite-state systems such as  Petri nets. Are there infinite-state restrictions of Petri nets other than 1-VASS whose visibly semantics is robust?

A direct application of our results is that the language $L \subseteq (\Sigma \times \N)^\ast$ of an OCA with observability semantics/strong automaton is learnable (in the sense of Angluin) in terms of a visibly one-counter automaton for $\enc{L}$ \cite{NeiderLoeding10}. It would be worthwhile to transfer results on visibly one-counter/pushdown automata that concern Myhill-Nerode congruences or minimization \cite{AlurKMV05,ChervetW07}.

Another interesting question is to which extent we can relax the requirement that the counter value be output with every letter $a \in \Sigma$. It may indeed be possible to deal with a bounded number of $\Sigma$-transitions between any two counter outputs. Note that there have been relaxations of the visibility condition in pushdown automata, albeit preserving closure under boolean operations \cite{NowotkaS07}.

\paragraph*{Acknowledgments.}

The author is grateful to C.\ Aiswarya, Stefan G{\"{o}}ller, Christoph Haase, and Arnaud Sangnier for numerous helpful discussions and pointers to the literature.

\bibliography{lit}

% ------------------------------------------------------------
% ------------------------------------------------------------
% ------------------------------------------------------------
% Appendix
% ------------------------------------------------------------
% ------------------------------------------------------------
% ------------------------------------------------------------

\appendix

\clearpage

% --------------------------------------------------------
% --------------------------------------------------------
% --------------------------------------------------------
% --------------------------------------------------------

\section{Proof of Lemma~\ref{lem:prod}}\label{app:prod}

Let $\A_1 = (Q_1,\Sigma,\init_1,F_1,\thr_1,\Delta_1)$ and $\A_2 = (Q_2,\Sigma,\init_2,F_2,\thr_2,\Delta_2)$ be OCAs. Then, the OCA $\A_1 \times \A_2$ is given as $(Q_1 \times Q_2,\Sigma,(\init_1,\init_2),F_1 \times F_2,\max\{\thr_1,\thr_2\},\Delta)$.
For all $(q_1,q_2),(q_1',q_2') \in Q_1 \times Q_2$, $k \in \setz{\max\{m_1,m_2\}}$, and $\sigma \in \Sigma \cup \Op$, the transition relation $\Delta$ contains
\[((q_1,q_2),k,\sigma,(q_1',q_2'))\]
iff $(q_1,\min\{k,m_1\},\sigma,q_1') \in \Delta_1$ and $(q_2,\min\{k,m_2\},\sigma,q_2') \in \Delta_2$.

\smallskip

Let $(q_1,q_2),(q_1',q_2') \in Q_1 \times Q_2$, $x \in \N$, and $\tau \in (\Sigma \times \{x\}) \cup \Op$. Define $\sigma \in \Sigma \cup \Op$ as follows: If $\tau \in \Op$, then $\sigma = \tau$, and if $\tau = (a,x) \in \Sigma \times \{x\}$, then $\sigma = a$.
Moreover, let $x' = x$ if $\sigma \in \Sigma$, $x' = x + 1$ if $\sigma = \inc$, and $x' = x - 1$ if $\sigma = \dec$.
We have:
\[
\begin{array}{rl}
& ((q_1,q_2),x) \xRightarrow{\;\tau\;}_{\!\A_1 \times \A_2} ((q_1',q_2'),x')\\[1ex]
\textup{iff} & ((q_1,q_2),\min\{x,\max\{m_1,m_2\}\},\sigma,(q_1',q_2')) \in \Delta\\[1ex]
\textup{iff} & (q_1,\min\{\min\{x,\max\{m_1,m_2\}\},m_1\},\sigma,q_1') \in \Delta_1 \textup{ and}\\
& (q_2,\min\{\min\{x,\max\{m_1,m_2\}\},m_2\},\sigma,q_2') \in \Delta_2\\[1ex]
\textup{iff} & (q_1,\min\{x,m_1\},\sigma,q_1') \in \Delta_1 \textup{ and }
(q_2,\min\{x,m_2\},\sigma,q_2') \in \Delta_2\\[1ex]
\textup{iff} & (q_1,x) \xRightarrow{\;\tau\;}_{\!\A_1} (q_1',x') \textup{ and }
(q_2,x) \xRightarrow{\;\tau\;}_{\!\A_2} (q_2',x')
\end{array}
\]
We conclude $\vL(\A_1 \times \A_2) = \vL(\A_1) \cap \vL(\A_2)$. Note that, if $\A_1$ and $\A_2$ are deterministic, then clearly so is $\A_1 \times \A_2$.

\section{Proof of Lemma~\ref{lem:benc}}\label{app:benc}

We define $\Benc=(Q,\Sigma,\init,F,0,\delta)$ by $Q = \{\incmode,\nmode,\decmode,\bot\}$, $\init = \nmode$, and $F = \{\nmode\}$. Moreover, for $q \in Q$ and $\sigma \in \Sigma \cup \Op$, we let
\[\delta(q,0,\sigma) =
\begin{cases}
\incmode & \text{if } q \in \{\incmode\,,\nmode\} \text{ and } \sigma = \inc\\
\decmode & \text{if } q \in \{\decmode\,,\nmode\} \text{ and }  \sigma = \dec\\
\nmode & \text{if } q \in \{\incmode\,,\nmode,\decmode\} \text{ and } \sigma \in \Sigma\\
\bot & \text{otherwise}
\end{cases}\]
It is easy to see that $\vL(\Benc) = \Encs{\Sigma}$.
In fact, $\Benc$ enters an ``increasing'' or ``decreasing'' mode as soon as it performs $\inc$ or, respectively, $\dec$. Such a mode can only be quit by reading a letter from $\Sigma$ or entering the sink state $\bot$. 
This avoids forbidden reversals between $\inc$ and $\dec$. As, moreover, the only final state is $\nmode$, any accepted word that is not empty ends in a letter from $\Sigma$.

\section{Proof of Lemma~\ref{lem:enc}}\label{app:enc}

We have to show $\vL(\A) = \enc{\oL(\A)}$.

\medskip

Towards ``\,$\subseteq$\,'', let $w = u_1a_1 \ldots u_na_n \in \vL(\A) \subseteq \Encs{\Sigma}$  where $a_i \in \Sigma$ and $u_i \in \{\inc\}^\ast \mathrel{\cup} \{\dec\}^\ast$ for all $i \in \set{n}$. There is an accepting run $\rho$ of $\A$ such that $\vextr(\trace(\rho)) = w$. By the definition of ${\srelp{\A}}$, we have $\oextr(\trace(\rho)) = (a_1,x_1) \ldots (a_n,x_n)$ where (letting $x_0 = 0$), for all $i \in \set{n}$, we have $x_i = x_{i-1} + |u_i|$ if $u_i \in \{\inc\}^\ast$, and $x_i = x_{i-1} - |u_i|$ if $u_i \in \{\dec\}^\ast$. We easily verify that
$\enc{\oextr(\trace(\rho))} = \vextr(\trace(\rho)) = w$. As $\oextr(\trace(\rho)) \in \oL(\A)$, we are done.

\medskip

To show ``\,$\subseteq$\,'', let $w=(a_1,x_1) \ldots (a_n,x_n) \in \oL(\A)$.
Since $\vL(\A) \subseteq \Encs{\Sigma}$,
there is an accepting run $\rho$ of $\A$ of the form
\[(q_0,\cv_0) \srelpp{\A}{u_1\;} \, \xRightarrow{(a_1,x_1)\,}_{\!\A} (q_1,\cv_1) \ldots \srelpp{\A}{u_n\;} \, \xRightarrow{(a_n,x_n)\,}_{\!\A} (q_n,\cv_n)\]
where $u_i \in \{\inc\}^\ast \cup \{\dec\}^\ast$ for all $i \in \set{n}$ (and with the expected meaning of $\srelpp{\A}{u_i\;}$). But this implies
that $u_i = \textup{sign}(x_i - x_{i-1})^{|x_i - x_{i-1}|}$ for all $i \in \set{n}$. Thus, $u_1a_1 \ldots u_na_n = \enc{w}$. As, moreover, $\vextr(\trace(\rho)) = u_1a_1 \ldots u_na_n \in \vL(\A)$, this concludes the proof.

\section{Proof of Proposition~\ref{prop:intersect}}\label{app:intersect}
We will show $\oL(\A_1 \times \A_2) = \oL(\A_1) \cap \oL(\A_2)$:
\[
\begin{array}{rcl}
\oL(\A_1 \times \A_2) & \stackrel{\textup{Lem.~\ref{lem:prod},\ref{lem:enc}}}{=} & \encmap^{-1}(\vL(\A_1 \times \A_2))\\
& \stackrel{\textup{Lem.~\ref{lem:prod}}}{=} & \encmap^{-1}(\vL(\A_1) \cap \vL(\A_2))\\
& \stackrel{\textup{Lem.~\ref{lem:enc}}}{=} & \encmap^{-1}(\enc{\oL(\A_1)} \cap \enc{\oL(\A_2)})\\[0.5ex]
& = & \encmap^{-1}(\enc{\oL(\A_1) \cap \oL(\A_2)})\\[0.5ex]
& = & \oL(\A_1) \cap \oL(\A_2)
\end{array}
\]

\section{Proof Details for Lemma~\ref{lem:eoca}}\label{app:eoca}

We will show
\[\oL(\eA{\A}) \stackrel{\textup{(a)}}{\subseteq} \oL(\A) \stackrel{\textup{(b)}}{\subseteq} \encmap^{-1}(\vL(\eA{\A}) \cap \Encs{\Sigma}) \stackrel{\textup{(c)}}{\subseteq} \oL(\eA{\A})\]

Case (c) is obvious.

\paragraph{(a)}

Let $w \in \oL(\eA{\A})$.
There is an accepting run
\[\rho' = (q_0,\cv_0) \srelpp{\eA{\A}}{\obsletter_1\;} (q_1,\cv_1) \srelpp{\eA{\A}}{\obsletter_2\;} \ldots \srelpp{\eA{\A}}{\obsletter_n\;} (q_n,\cv_n)\]
of $\eA{\A}$ such that $w = \oextr(\trace(\rho'))$. We will
construct an accepting run $\rho$
of $\A$ such that $\oextr(\trace(\rho)) = \oextr(\trace(\rho'))$. This proves $w \in \oL(\A)$. Pick $i \in \set{n}$.
We have $(q_{i-1},\cv_{i-1}) \srelpp{\eA{\A}}{\obsletter_i\;} (q_i,\cv_i)$. By the definition of $\Delta'$, there exist $\hat{q}_{i-1} \in Q$, $k_{i-1} \in \setz{\thr}$, and $\sigma_i \in \Sigma \cup \Op$ such that
\begin{enumerate}\itemsep=0.5ex
\item $x_{i-1} \models \phi_{q_{i-1},\hat{q}_{i-1}} \wedge \kform{k_{i-1}}$,
\item $(\hat{q}_{i-1},k_{i-1},\sigma_i,q_i) \in \Delta$, and
\item one of the following holds:
\begin{itemize}\itemsep=0.5ex
\item $\tau_i = (a,x) \in \Sigma \times \N$ and $x_i = x_{i-1} = x$ and $\sigma_i = a$,
\item $\tau_i = \inc$ and $x_i = x_{i-1} + 1$ and $\sigma_i = \inc$, or
\item $\tau_i = \dec$ and $x_i = x_{i-1} - 1$ and $\sigma_i = \dec$.
\end{itemize}
\end{enumerate}
Note that 1.\ implies $k_{i-1} = \min\{x_{i-1},\thr\}$. By Lemma~\ref{lem:rpq}, we get
\[(q_{i-1},\cv_{i-1}) \mathrel{\smash[b]{\stackrel{\epsilon}{\Longrightarrow}}_{\!\A}^\ast} (\hat{q}_{i-1},\cv_{i-1}) \srelpp{\A}{\obsletter_i\;} (q_i,\cv_i)\]
where we set ${\smash[b]{\stackrel{\epsilon}{\Longrightarrow}}_{\!\A}} = {\srelpp{\A}{\sinc\;} \cup \srelpp{\A}{\sdec\;}}$. We have $x_n \models \bigvee_{q \in F} \psi_{q_n,q}$.
Thus, by Lemma~\ref{lem:rpq},
\[(q_n,\cv_n) \mathrel{\smash[b]{\stackrel{\epsilon}{\Longrightarrow}}_{\!\A}^\ast} (q,\cv)\] for some $q \in F$ and $\cv \in \N$.
Note that $(q,x)$ is a final configuration of $\A$.
Putting everything together, we obtain an accepting run
\[
\begin{array}{rll}
\rho = (q_0,\cv_0) & \mathrel{\smash[b]{\stackrel{\epsilon}{\Longrightarrow}}_{\!\A}^\ast} (\hat{q}_{0},\cv_{0})
& \srelpp{\A}{\;\obsletter_1\;} (q_1,\cv_1)\\
&  \mathrel{\smash[b]{\stackrel{\epsilon}{\Longrightarrow}}_{\!\A}^\ast} (\hat{q}_{1},\cv_{1})
& \srelpp{\A}{\;\obsletter_2\;} (q_{2},\cv_{2})\\
& \qquad\qquad\qquad\qquad \vdots\\
&  \mathrel{\smash[b]{\stackrel{\epsilon}{\Longrightarrow}}_{\!\A}^\ast} (\hat{q}_{n-1},\cv_{n-1})
& \srelpp{\A}{\obsletter_n\;} (q_{n},\cv_{n})
 \mathrel{\smash[b]{\stackrel{\epsilon}{\Longrightarrow}}_{\!\A}^\ast} (q,\cv)
\end{array}
\]
of $\A$. Here, the star operation $\ast$ stands for an arbitrary but concrete number of repetitions. Obviously, $\oextr(\trace(\rho)) = \oextr(\tau_1 \ldots \tau_n) = \oextr(\trace(\rho'))$ so that we are done.

\paragraph{(b)}

\newcommand{\qst}[2]{p_{#1,#2}}
\newcommand{\qstp}[2]{q_{#1,#2}}
\newcommand{\qstpp}[2]{q_{#1,#2}'}

Next, we show that $\oL(\A) \subseteq \encmap^{-1}(\vL(\eA{\A}) \cap \Encs{\Sigma})$.
Let $w = (a_1,x_1) \ldots (a_n,x_n) \in \oL(\A)$.
There is an accepting run $\rho$ of $\A$ such that $w=\oextr(\trace(\rho))$.
The run $\rho$ can be decomposed as
\[
\begin{array}{rll}
(q_0,x_0) & \left(\mathrel{\smash[b]{\stackrel{\epsilon}{\Longrightarrow}}_{\!\A}^\ast}\; \mathrel{\xRightarrow{\op_1\,}_{\!\A}}\right)^{x_1}
& \mathrel{\smash[b]{\stackrel{\epsilon}{\Longrightarrow}}_{\!\A}^\ast}\; \xRightarrow{(a_1,x_1)\,}_{\!\A} (q_1,x_1)\\
& \left(\mathrel{\smash[b]{\stackrel{\epsilon}{\Longrightarrow}}_{\!\A}^\ast}\; \mathrel{\xRightarrow{\op_2\,}_{\!\A}}\right)^{|x_2 - x_1|}
& \mathrel{\smash[b]{\stackrel{\epsilon}{\Longrightarrow}}_{\!\A}^\ast}\; \xRightarrow{(a_2,x_2)\,}_{\!\A} (q_2,x_2)\\
& \qquad\qquad\qquad \vdots\\
& \left(\mathrel{\smash[b]{\stackrel{\epsilon}{\Longrightarrow}}_{\!\A}^\ast}\; \mathrel{\xRightarrow{\op_n\,}_{\!\A}}\right)^{|x_n - x_{n-1}|}
& \mathrel{\smash[b]{\stackrel{\epsilon}{\Longrightarrow}}_{\!\A}^\ast}\; \xRightarrow{(a_n,x_n)\,}_{\!\A} (q_n,x_n)
\mathrel{\smash[b]{\stackrel{\epsilon}{\Longrightarrow}}_{\!\A}^\ast} (q_n',x_n')
\end{array}
\]
where $\op_i = \textup{sign}(x_{i} - x_{i-1})$ (as defined in Section~\ref{sec:obs-vis}) and where, for all $i \in \set{n}$, the partial run
\[(q_{i-1},x_{i-1})
\left(\mathrel{\smash[b]{\stackrel{\epsilon}{\Longrightarrow}}_{\!\A}^\ast}\; \mathrel{\xRightarrow{\op_i\,}_{\!\A}}\right)^{|x_i - x_{i-1}|}
\mathrel{\smash[b]{\stackrel{\epsilon}{\Longrightarrow}}_{\!\A}^\ast}\; \xRightarrow{(a_i,x_i)\,}_{\!\A} (q_i,x_i)\]
is actually of the form
\[
\begin{array}{rll}
(q_{i-1},x_{i-1}) = (p_0,y_0)
& \mathrel{\smash[b]{\stackrel{\epsilon}{\Longrightarrow}}_{\!\A}^\ast} (r_0,y_0)
& \mathrel{\xRightarrow{\op_i\,}_{\!\A}} (p_1,y_1)\\[1ex]
& \mathrel{\smash[b]{\stackrel{\epsilon}{\Longrightarrow}}_{\!\A}^\ast} (r_1,y_1)
& \mathrel{\xRightarrow{\op_i\,}_{\!\A}} (p_2,y_2)\\
& \qquad\qquad\vdots\\
& \mathrel{\smash[b]{\stackrel{\epsilon}{\Longrightarrow}}_{\!\A}^\ast} (r_{\ell-1},y_{\ell-1})
& \mathrel{\xRightarrow{\op_i\,}_{\!\A}} (p_\ell,y_\ell)\\[1ex]
& \mathrel{\smash[b]{\stackrel{\epsilon}{\Longrightarrow}}_{\!\A}^\ast} (r_\ell,y_\ell)
& \xRightarrow{(a_i,x_i)\,}_{\!\A} (q_i,x_i)
\end{array}
\]
where $\ell = |x_i - x_{i-1}|$ and $y_j = x_{i-1} + j \cdot \op_i$ (with the expected meaning).

Pick $j \in \{1,\ldots,\ell\}$. Since
\[(p_{j-1},y_{j-1})
\mathrel{\smash[b]{\stackrel{\epsilon}{\Longrightarrow}}_{\!\A}^\ast} (r_{j-1},y_{j-1})
\mathrel{\xRightarrow{\op_i\,}_{\!\A}} (p_j,y_j)\,,\] we have
$y_{j-1} \models \phi_{p_{j-1},r_{j-1}}$ and
$(r_{j-1},\min\{y_{j-1},\thr\},\op_i,p_j) \in \Delta$.
This implies $(p_{j-1},\phi_{p_{j-1},r_{j-1}} \wedge \kform{\min\{y_{j-1},\thr\}},\op_i,p_j) \in \Delta'$. Therefore,
\[(p_{j-1},y_{j-1}) \mathrel{\xRightarrow{\op_i\,}_{\!\eA{\A}}} (p_{j},y_j)\,.\]
Letting $j$ range over $\{1,\ldots,\ell\}$,
we obtain a partial run
\[\rho_i \df (q_{i-1},x_{i-1}) \left(\mathrel{\xRightarrow{\op_i\,}_{\!\eA{\A}}}\right)^{|x_i - x_{i-1}|}
(p_\ell,y_\ell)\]
of $\eA{\A}$.
Similarly, by
\[(p_\ell,y_\ell)
\mathrel{\smash[b]{\stackrel{\epsilon}{\Longrightarrow}}_{\!\A}^\ast} (r_\ell,y_\ell)
\xRightarrow{(a_i,x_i)\,}_{\!\A} (q_i,x_i)\,,\] we obtain
$y_\ell \models \phi_{p_\ell,r_\ell}$ and
$(r_\ell,\min\{y_\ell,\thr\},a_i,q_i) \in \Delta$.
Thus, $(p_\ell,\phi_{p_\ell,r_\ell} \wedge \kform{\min\{y_\ell,\thr\}},a_i,q_i) \in \Delta'$, which implies
\[{\rho}_i' \df (p_\ell,y_\ell) \xRightarrow{(a_i,x_i)\,}_{\!\eA{\A}} (q_i,x_i)\,.\]

\smallskip

Finally, concatenating $\rho_1,{\rho}_1',\ldots,\rho_n,{\rho}_n'$ yields a run $\rho'$ of $\eA{\A}$.  Since $\rho$ is accepting, we have $q_n' \in F$. By $(q_n,x_n)
\mathrel{\smash[b]{\stackrel{\epsilon}{\Longrightarrow}}_{\!\A}^\ast} (q_n',x_n')$, we also have $x_n \models \psi_{q_n,q_n'}$. Thus, $(q_n,x_n)$ is a ``final configuration'' of $\eA{\A}$ so that $\rho'$ is accepting.

Obviously, $\oextr(\trace(\rho)) = \oextr(\trace(\rho'))$.
By the definition of $\encmap$, we also have $\vextr(\trace(\rho')) = \enc{w}$.
We conclude that $w \in \encmap^{-1}(\vL(\eA{\A}) \cap \Encs{\Sigma})$.

\section{Proof of Lemma~\ref{lem:BC}}\label{app:BC}

Let $\SetGuards \subseteq \Guardsmod$ be a nonempty finite set and let $\thr_\SetGuards = \max\{\consta \mid \modform{\consta}{\constb}$ is an atomic subformula of some $\phi \in \SetGuards\}+2$. Suppose $\{\modform{\consta_1}{\constb_1},\ldots,\modform{\consta_n}{\constb_n}\}$ is the set of all atomic formulas $\modform{\consta}{\constb}$ that occur as subformulas in $\SetGuards$.
Note that $m_\SetGuards = \max\{c_1,\ldots,c_n\}+2$. We construct $\B_\SetGuards=(Q,\Sigma,\init,Q,\thr_\SetGuards,\delta)$ and the mapping $\eval$ as follows:
\begin{itemize}\itemsep=0.5ex
\item $Q = \prod_{i \in \set{n}} (\{0,\ldots,\consta_i\} \times \{0,\ldots,\constb_i'-1\})$ where $\constb_i' = \max\{1,\constb_i\}$

\item $\init = ((0,0),\ldots,(0,0))$

\item $\delta((q_1,\ldots,q_n),k,\sigma) =
\begin{cases}
(q_1,\ldots,q_n) & \text{if } \sigma \in \Sigma\\
(\delta_1(q_1,k,\sigma),\ldots,\delta_n(q_n,k,\sigma)) & \text{if } \sigma \in \{\inc,\dec\}
\end{cases}$
\end{itemize}
Here, for $i \in \set{n}$, $(x,y) \in \setz{\consta_i} \times \setz{\constb_i'-1}$, $k \in \setz{\thr_\SetGuards}$, and $\sigma \in \Op$,
\[\begin{array}{l}
\begin{array}{lcl}
\delta_i((x,y),k,\sigma) & = &
\left.
\begin{cases}
(c_i,0) & \text{if }
\left(\begin{array}{rl}
& \sigma = \inc \,\wedge\, k = c_i-1\\
\vee & \sigma = \dec \,\wedge\, k = c_i+1
\end{array}\right)\\
(c_i,1) & \text{otherwise} 
\end{cases}~~~\right| \text{ if } d_i = 0
\end{array}\\~\\
\left.
\begin{array}{lcl}
\delta_i((x,y),k,\inc) & = &
\begin{cases}
(x,(y+1) \mathrel{\textup{mod}} d_i) & \text{if } k \ge c_i\\
(x+1,0) & \text{if } k < c_i
\end{cases}\\~\\
\delta_i((x,y),k,\dec) & = &
\begin{cases}
(x,(y-1) \mathrel{\textup{mod}} d_i) & \text{if } k > c_i+1\\
(\max\{0,x-1\},0) & \text{if } k \le c_i+1
\end{cases}
\end{array}~~~\right| \text{ if } d_i \ge 1
\end{array}\]
with $-1 \mathrel{\textup{mod}} d_i = d_i-1$.

Let us define the mapping $\eval$. Let $\bar{q} = ({q}_1,\ldots,{q}_n) \in Q$.
For $i \in \set{n}$, we write $\bar{q} \models \modform{c_i}{d_i}$ iff ${q}_i = (c_i,0)$.
Moreover, given $\phi \in \SetGuards$, we write $\bar{q} \models \phi$ iff $\phi$ is evaluated to true given the truth values of $\bar{q} \models \modform{c}{d}$ for the atomic subformulas $\modform{c}{d}$ of $\phi$. With this, we set $\eval(\bar{q}) = \{\phi \in \SetGuards \mid \bar{q} \models \phi\}$.

\section{Proof Details for Lemma~\ref{thm:oca-to-doca}}\label{app:oca-to-doca}

Let $\extr \in \{\cextr,\vextr,\oextr\}$. We first show $\eL(\A') \subseteq \eL(\A)$. Suppose $w  \in \eL(\A')$. There is an accepting run
\[\rho' = ((P_0,q_0),\cv_0) \srelpp{\A'}{\obsletter_1\;} ((P_1,q_1),\cv_1) \srelpp{\A'}{\obsletter_2\;} \ldots \srelpp{\A'}{\obsletter_n\;} ((P_n,q_n),\cv_n)\]
such that $w = \extr(\tau_1 \ldots \tau_n)$. We will determine $p_0,p_1,\ldots,p_n \in Q$ such that
\[\rho = (p_0,\cv_0) \srelpp{{\A}}{\obsletter_1\;} (p_1,\cv_1) \srelpp{{\A}}{\obsletter_2\;} \ldots \srelpp{{\A}}{\obsletter_n\;} (p_n,\cv_n)\]
is an accepting run of ${\A}$. This will imply $w = \extr(\tau_1 \ldots \tau_n) \in \eL(\A)$.

As $\rho'$ is accepting, we have $(P_n,q_n) \in F'$. That is, there is $p_n \in P_n$ such that $\fmap(p_n) \in \eval(q_n)$.
The latter implies $x_n \models \fmap(p_n)$ so that $(p_n,x_n)$ is a final configuration of $\A$.

Now, suppose that we already have $p_i \in P_i$ for some $i \in \set{n}$. We find $p_{i-1}$ as follows.
If $\tau_i \in \Op$, set $\sigma_i = \tau_i$, and if $\tau_i = (a,x_i)$ for some $a \in \Sigma$, then set $\sigma_i = a$.
By $((P_{i-1},q_{i-1}),\cv_{i-1}) \srelpp{\A'}{\obsletter_i\;} ((P_i,q_i),\cv_i)$,
we have
\[\delta'((P_{i-1},q_{i-1}),k,\sigma_i) = (P_i,\hat\delta(q_{i-1},k,\sigma_i))\] where $k = \min\{x_{i-1},\thr_\SetGuards\}$.
Thus, there are $p_{i-1} \in P_{i-1}$ and $\phi \in \eval(q_{i-1})$ such that $(p_{i-1},\phi,\sigma_i,p_i) \in \Delta$ is a transition of $\A$. Since $\phi \in \eval(q_{i-1})$ implies $x_{i-1} \models \phi$, we have $(p_{i-1},\cv_{i-1}) \srelpp{\A}{\obsletter_i\;} (p_i,\cv_i)$. We proceed in that way until we found $p_0,p_1,\ldots,p_n$.

\medskip

Conversely, let us show $\eL(\A) \subseteq \eL(\A')$. Take any word
$w \in \eL(\A)$. There is an accepting run
\[\rho = (p_0,\cv_0) \srelpp{\A}{\obsletter_1\;} (p_1,\cv_1) \srelpp{\A}{\obsletter_2\;} \ldots \srelpp{\A}{\obsletter_n\;} (p_n,\cv_n)\]
of $\A$ such that $w = \extr(\tau_1 \ldots \tau_n)$. Since $\A'$ is deterministic, there is a unique run of $\A'$ of the form
\[\rho' = ((P_0,q_0),\cv_0) \srelpp{\A'}{\obsletter_1\;} ((P_1,q_1),\cv_1) \srelpp{\A'}{\obsletter_2\;} \ldots \srelpp{\A'}{\obsletter_n\;} ((P_n,q_n),\cv_n)\]
We will show that $\rho'$ is accepting so that we are done.

To do so, we show that $p_i \in P_i$ for all $i \in \{0,\ldots,n\}$.
Of course, $p_0 = \init \in \{\init\} = P_0$.
From $p_{i-1} \in P_{i-1}$ (with $i \ge 1$), we will now infer $p_i \in P_i$.
Let $\sigma_i$ be defined as above.
We have $(p_{i-1},\cv_{i-1}) \srelpp{\A}{\obsletter_i\;} (p_i,\cv_i)$.
Therefore, there is $\phi \in \Guardsmod$ such that
$(p_{i-1},\phi,\sigma_i,p_i) \in \Delta$ is a transition of $\A$
and $x_{i-1} \models \phi$. The latter implies $\phi \in \eval(q_{i-1})$ so that,
we have $p_i \in P_i$.

We finally obtain $p_n \in P_n$. Recall that $F' = \{(P,q) \in Q' \mid \fmap(p) \in \eval(q)$ for some $p \in P\}$.
Since $x_n \models \fmap(p_n)$, we have $\fmap(p_n) \in \eval(q_n)$. Thus, $(P_n,q_n) \in F'$ so that $\rho'$ is indeed accepting.

\section{Proof Details for Theorem~\ref{thm:strong-aut}}\label{app:strong-aut}

\newcommand{\T}{\mathcal{T}}

``$\Longrightarrow$'':\; Given $q,q'$, we transform $\A$ into a two-way automaton $\T_{q,q'}$ over infinite words \cite{Pecuchet85} over the alphabet $2^{\{\$,\$'\}}$. The idea is that word positions represent counter values (the first position marking $0$, the second $1$, and so on) and $\$$ and $\$'$ represent $x$ and $x'$, respectively. Thus, we are only interested in words in which $\$$ and $\$'$ each occur exactly once. Clearly, this is a regular property. At the beginning, $\T_{q,q'}$ goes to the position carrying $\$$. It then simulates $\A$ starting in $q$, and it accepts if it is on the position carrying $\$'$ and in state $q'$. The simulation itself is straightforward: Counter increments and decrements of an OCA translate to going to the left and to the right, respectively, and a zero test simply checks whether the automaton is at the first position of the word. Note that $\T_{q,q'}$ checks for the markers $\$$ and $\$'$ only at the beginning and at the end of an execution, but not during the actual simulation of $\A$.
Let $w=z_0z_1z_2 \ldots \in \bigl(2^{\{\$,\$'\}}\bigr)^\omega$ and $x,x' \in \N$ be unique positions such that $\$ \in z_x$ and $\$' \in z_{x'}$. Then $w$ is accepted by $\T_{q,q'}$ iff $(q,\cv) \mathrel{{(\srelpp{\A}{\sinc\;} \cup \srelpp{\A}{\sdec\;})^\ast}} (q',\cv')$.

It is well known that two-way word automata are expressively equivalent to one-way automata. Therefore, by B{\"u}chi's theorem, the word language accepted by $\T_{q,q'}$ can be translated into a corresponding MSO formula without free variables but with subformulas of the form ``position $\mathsf{y}$ carries $\$$'' and ``position $\mathsf{y}$ carries $\$'$\,'' \cite{Tho97handbook}. We replace the latter two by $\mathsf{y} = \varx$ and $\mathsf{y} = \varx'$, respectively, and finally obtain $\transform_{q,q'}$ as required.

\bigskip

\noindent ``$\Longleftarrow$'':\; It remains to show $\oL(\A) = L(\SA)$.

\medskip

Towards ``$\supseteq$'', let
$w = (a_1,x_1) \ldots (a_n,x_n) \in L(\SA)$.
There is an accepting run
\[\rho = (p_0,x_0) \xRightarrow{(a_1,x_1)\,}_{\SA} (p_1,\cv_1) \xRightarrow{(a_2,x_2)\,}_{\SA} \ldots \xRightarrow{(a_n,x_n)\,}_{\SA} (p_n,\cv_n)\]
of $\SA$ on $w$.
There are $\transform_1,\ldots,\transform_{n}$ such that
$(p_{i-1},\transform_{i},a_i,p_i) \in \Delta$ and $(x_{i-1},x_i) \models \transform_{i}$ for all $i \in \set{n}$. In addition, we set $\transform_{n+1} = \false$.
Moreover, we set $r_i = \delta_{\transform_i}(\init_{\transform_i},0,\inc^{x_{i-1}})$
(which is defined as expected)
and $q_i = \delta_{\transform_i}(\init_{\transform_i},0,\inc^{x_{i-1}}\markerone)$
for all $i \in \set{n+1}$. Note that $r_1 = \init_{\transform_1}$ and $q_1 = \delta_{\transform_1}(\init_{\transform_1},0,\markerone)$.

Pick $i \in \set{n}$.
Suppose  $x_{i-1} \ge x_i$. Then, $\inc^{x_{i-1}} \markerone \dec^{x_{i-1} - x_{i}} \markertwo \in L(\C_{\transform_i})$. Therefore, there is an accepting run of $\C_{\transform_i}$ of the form 
\[(\init_{\transform_i},0) \left(\xRightarrow{\;\sinc\;}_{\C_{\transform_i}}\right)^{x_{i-1}} (r_{i},x_{i-1}) \xRightarrow{\,(\markerone,x_{i-1})\,}_{\C_{\transform_i}} (q_{i},x_{i-1}) \left(\xRightarrow{\;\sdec\;}_{\C_{\transform_i}}\right)^\ell (q_i',x_i) \xRightarrow{\,(\markertwo,x_i)\,}_{\C_{\transform_i}} (\hat{q}_i,x_i)\]
where $\ell = x_{i-1} - x_i$ and $\hat{q}_i \in F_{\transform_i}$. Thus, $\A$ exhibits a sequence of transitions
\[(p_{i-1},\transform_{i},q_{i}) \xrightarrow{\,\pi_{k_0}\,|\,\sdec\,} \ldots \xrightarrow{\,\pi_{k_{\ell-1}}\,|\,\sdec\,} (p_{i-1},\transform_{i},q_i') \xrightarrow{\,\incform_{r_{i+1}} \wedge \pi_{k_\ell}\,|\,a_i\,} (p_i,\transform_{i+1},q_{i+1})\]
where $k_j = {\min\{x_{i-1}-j,m\}}$ for all $j \in \setz{\ell}$.
By $r_{i+1} = \delta_{\transform_{i+1}}(\init_{\transform_{i+1}},0,\inc^{x_{i}})$, we have $x_i \in \sem{\incform_{r_{i+1}}}$. Thus,
\[((p_{i-1},\transform_{i},q_{i}),x_{i-1}) \left(\srelpp{\A}{\sdec\;}\right)^\ell
((p_{i-1},\transform_{i},q_i'),x_i)
\xRightarrow{(a_i,x_i)\,}_{\!\A} ((p_i,\transform_{i+1},q_{i+1}),x_i)\,.
\]
The case $x_{i-1} < x_i$ is handled analogously.

Altogether, we obtain
\[((p_0,\transform_{1},q_{1}),x_0) \mathrel{\smash[b]{\stackrel{\epsilon}{\Longrightarrow}}_{\!\A}^\ast}\; \xRightarrow{(a_1,x_1)\,}_{\!\A}
\ldots  \mathrel{\smash[b]{\stackrel{\epsilon}{\Longrightarrow}}_{\!\A}^\ast}\; \xRightarrow{(a_n,x_n)\,}_{\!\A}
((p_n,\transform_{n+1},q_{n+1}),x_n)
\]
where
${\smash[b]{\stackrel{\epsilon}{\Longrightarrow}}_{\!\A}} = {\srelpp{\A}{\sinc\;} \cup \srelpp{\A}{\sdec\;}}$.
Note that $(p_0,\transform_{1},q_{1})$ is an initial state, since $p_0 = \init$ and $q_1 = \delta_{\transform_1}(\init_{\transform_1},0,\inc^0 \markerone)$.
Moreover, $((p_n,\transform_{n+1},q_{n+1}),x_n)$ is a final configuration, since $p_n \in F$ and $\transform_{n+1} = \false$.
We obtain $(a_1,x_1) \ldots (a_n,x_n) \in \oL(\A)$.

\medskip

Towards ``$\subseteq$'', suppose $(a_1,x_1) \ldots (a_n,x_n) \in \oL(\A)$.
We first observe that $\vL(\A) \subseteq \Encs{\Sigma}$. This is by the languages $L(\C_\transform)$ as well as the definition of $\Delta'$. In fact, there is an accepting run of $\A$ of the form
\[
\begin{array}{rlll}
((p_0,\transform_1,q_1),x_0)  & 
\left(\mathrel{\xRightarrow{\;\sinc\;}_{\!\A}}\right)^{x_1} &
((p_0,\transform_1,q_1'),x_1) &
\mathrel{\xRightarrow{(a_1,x_1)\,}_{\!\A}}
((p_1,\transform_2,q_2),x_1)\\
& \left(\mathrel{\xRightarrow{\;\op_2\;}_{\!\A}}\right)^{|x_2-x_1|} &
((p_1,\transform_2,q_2'),x_2) &
\mathrel{\xRightarrow{(a_2,x_2)\,}_{\!\A}}
((p_2,\transform_3,q_3),x_2)\\
& & \vdots\\
((p_{i-1},\transform_i,q_i),x_{i-1}) & \left(\mathrel{\xRightarrow{\;\op_i\;}_{\!\A}}\right)^{|x_i-x_{i-1}|} &
((p_{i-1},\transform_i,q_i'),x_i) &
\mathrel{\xRightarrow{(a_i,x_i)\,}_{\!\A}}
((p_i,\transform_{i+1},q_{i+1}),x_i)\\
& & \vdots\\
& \left(\mathrel{\xRightarrow{\;\op_n\;}_{\!\A}}\right)^{|x_n-x_{n-1}|} &
((p_{n-1},\transform_n,q_n'),x_n) &
\mathrel{\xRightarrow{(a_n,x_n)\,}_{\!\A}}
((p_n,\transform_{n+1},q_{n+1}),x_n)
\end{array}
\]
where $\transform_{n+1} = \false$, $\op_i = \textup{sign}(x_i-x_{i-1})$ for all $i \in \{2,\ldots,n\}$, and $(p_{i-1},\transform_i,a_i,p_i) \in \Delta$ for all $i \in \set{n}$.
We will show that
\[\rho = (p_0,x_0) \xRightarrow{(a_1,x_1)\,}_{\SA} (p_1,\cv_1) \xRightarrow{(a_2,x_2)\,}_{\SA} \ldots \xRightarrow{(a_n,x_n)\,}_{\SA} (p_n,\cv_n)\]
is an accepting run of $\SA$.
Pick $i \in \set{n}$. It only remains to show $(x_{i-1},x_i) \models \transform_i$.
By the definition of $\Delta'$, there is $r_i \in Q_{\transform_i}$ such that
$q_i = \delta_{\transform_i}(r_i,\min\{x_{i-1},m\},\markerone)$ and $x_{i-1} \in \sem{\incform_{r_i}}$.
Thanks to the definition of $\incform_{r_i}$ and $\Delta'$, there is an accepting run of $\C_{\transform_i}$ of the form
\[(\init_{\transform_i},0) \left(\xRightarrow{\;\sinc\;}_{\C_{\transform_i}}\right)^{x_{i-1}} (r_{i},x_{i-1}) \xRightarrow{\,(\markerone,x_{i-1})\;}_{\C_{\transform_i}} (q_{i},x_{i-1}) \left(\xRightarrow{\;\op_i\;}_{\C_{\transform_i}}\right)^{\ell_i} (q_i',x_i) \xRightarrow{\,(\markertwo,x_i)\;}_{\C_{\transform_i}} (\hat{q}_i,x_i)\]
where $\ell_i = |x_i - x_{i-1}|$. But this implies
$\inc^{x_i} \markerone (\op_i)^{\ell_i} \markertwo \in \vL(\C_{\transform_i})$ and, by Claim~\ref{cl:Cphi},
$(x_{i-1},x_i) \models \transform_i$. Since $p_0 = \init$ and $p_n \in F$, we obtain that $\rho$ is indeed an accepting run.
We conclude $(a_1,x_1) \ldots (a_n,x_n) \in L(\SA)$.

\end{document}